\documentclass[sigplan,screen]{acmart}
\AtBeginDocument{%
  \providecommand\BibTeX{{%
    \normalfont B\kern-0.5em{\scshape i\kern-0.25em b}\kern-0.8em\TeX}}}

\copyrightyear{2021}
\acmYear{2021}
\setcopyright{acmlicensed}\acmConference[SOSP '21]{ACM SIGOPS 28th Symposium on Operating Systems Principles}{October 26--29, 2021}{Virtual Event, Germany}
\acmBooktitle{ACM SIGOPS 28th Symposium on Operating Systems Principles (SOSP '21), October 26--29, 2021, Virtual Event, Germany}
\acmPrice{15.00}
\acmDOI{10.1145/3477132.3483582}
\acmISBN{978-1-4503-8709-5/21/10}

\acmSubmissionID{617}

\usepackage{paralist}
\usepackage{multicol}
\usepackage{hyperref}
\usepackage{comment}

\usepackage{url}
\usepackage{algorithm,amsmath,amsthm,amsfonts}
\usepackage{thmtools,thm-restate}
\usepackage[normalem]{ulem}
\usepackage{listings}
\usepackage{lineno}

\usepackage[noend]{algorithmic}

\usepackage{graphicx}

\usepackage{tikz,subcaption}

\usepackage{soul}

\usepackage{pdfpages}

\newcommand{\name}{Rabia}

\usepackage[shortlabels]{enumitem}
\usepackage{svg}

\newtheorem{theorem}{Theorem}
\newtheorem{lemma}{Lemma}

\begin{document}

\title{Rabia: Simplifying State-Machine Replication Through Randomization}

\author{Haochen Pan}
\authornote{Work was done at Boston College.}
\affiliation{%
  \institution{University of Chicago}
  \city{Chicago}
  \state{IL}
  \country{USA}}
\email{haochenpan@uchicago.edu}

\author{Jesse Tuglu}
\authornotemark[1]
\affiliation{%
  \institution{University of Michigan}
  \city{Ann Arbor}
  \state{MI}
  \country{USA}}
\email{tuglu@umich.edu}

\author{Neo Zhou}
\affiliation{%
  \institution{Boston College}
  \city{Boston}
  \state{MA}
  \country{USA}}
\email{neo.zhou@bc.edu}

\author{Tianshu Wang}
\authornotemark[1]
\affiliation{%
 \institution{Duke University}
 \city{Durham}
 \state{NC}
 \country{USA}}
\email{tw295@duke.edu}

\author{Yicheng Shen}
\affiliation{%
  \institution{Boston College}
  \city{Boston}
  \state{MA}
  \country{USA}}
\email{yicheng.shen@bc.edu}

\author{Xiong Zheng}
\authornote{Work was done at UTexas Austin.}
\affiliation{%
  \institution{Google, Inc.}
  \city{Kirkland}
  \state{WA}
  \country{USA}}
\email{xiongzheng@google.com}

\author{Joseph Tassarotti}
\affiliation{%
  \institution{Boston College}
  \city{Boston}
  \state{MA}
  \country{USA}}
\email{tassarot@bc.edu}

\author{Lewis Tseng}
\affiliation{%
  \institution{Boston College}
  \city{Boston}
  \state{MA}
  \country{USA}}
\email{lewis.tseng@bc.edu}

\author{Roberto Palmieri}
\affiliation{%
  \institution{Lehigh University}
  \city{Bethlehem}
  \state{PA}
  \country{USA}}
\email{palmieri@lehigh.edu}

\renewcommand{\shortauthors}{Pan et al.}

\begin{abstract}
We introduce \name{}, a simple and high performance framework for implementing state-machine replication (SMR) within a datacenter. The main innovation of \name{} is in using \textit{randomization} to simplify the design.  \name{} provides the following two features: (i) It does not need any fail-over protocol and supports trivial auxiliary protocols like log compaction, snapshotting, and reconfiguration, components that are often considered the most challenging when developing SMR systems; and (ii) It provides high performance, up to 1.5x higher throughput than the closest competitor (i.e., EPaxos) in a favorable setup (same availability zone with three replicas) and is comparable with a larger number of replicas or when deployed in multiple availability zones.
\end{abstract}



\maketitle

\section{Introduction}
\label{s:intro}

State-Machine Replication (SMR) uses replication to ensure that a service is available and consistent in the presence of failures. One popular mechanism for implementing SMR is to use a consensus algorithm to agree on the total order of client requests (or commands), namely, a log-based SMR approach~\cite{Raft_ATC14,SMR_Schneider_ACM90,Mu_Aguilera_OSDI20}. Paxos and variants \cite{lamport1998part,lamport2001paxos} had mostly been the de facto choice for implementing SMR, e.g., Chubby \cite{Chubby_OSDI06}, Google Spanner \cite{Spanner_OSDI12}, Microsoft Azure Storage \cite{Azure_SOSPI11}. Raft \cite{Raft_ATC14} recently became a popular alternative, designed with understandability as the priority. Many modern production systems choose Raft over Paxos, e.g., Redis~\cite{RedisRaft}, RethinkDB~\cite{RethinkDBRaft}, CockroachDB~\cite{CockroachDBRaft}, and etcd \cite{etcd}. In this work, we aim to address one of the remaining challenges in building a \textit{high-performance log-based SMR} system within a single datacenter: reducing the engineering effort and simplifying the integration of auxiliary protocols.

Paxos and variants are notoriously difficult to understand and integrate with SMR \cite{Paxos_live_PODC07,Raft_ATC14}. Raft is sometimes considered easier to comprehend with its use of a stronger notion of leader; however,
it still requires a tremendous amount of effort to develop a fully functional Raft-based SMR system \cite{Raft_MITBlog16,Jepsen_Raft_Redis,Jepsen_Raft_RethinkDB16}. 
One root cause is that prior practical consensus algorithms, including Paxos, Raft, and recent systems, require auxiliary protocols (e.g., leader election, snapshotting,  fail-over/recovery mechanism, log compaction/truncation) for ensuring high performance and liveness. We will elaborate on these challenges in Section \ref{s:prior-SMR-challange}. 

\name{} is a \textit{simple} SMR framework that ensures a total order of client requests and achieves high performance in a single datacenter setting. \name{} does \textit{not} need any fail-over protocol and supports trivial log compaction, reconfiguration, and snapshotting protocols. We accomplish this through a novel implementation of consensus that leverages a \underline{RA}ndomized \underline{BI}nary \underline{A}greement protocol.\footnote{These algorithms are also called agreement algorithms, but we will refer to them as consensus algorithms.} 

Randomized algorithms are known to have less stable performance compared to deterministic ones. By limiting our focus to SMR systems deployed in a datacenter where the network infrastructure is \textit{stable}, we can achieve higher throughput than SMR systems implementing total order through Multi-Paxos \cite{lamport1998part,lamport2001paxos}, and Egalitarian Paxos (EPaxos) \cite{EPaxos_SOSP13} in a favorable setup (same availability zone with three replicas), and have comparable performance when deployed with a larger number of replicas or in multiple zones.

We also present RedisRabia, an integration of Rabia and Redis, a popular distributed in-memory data store. We compare RedisRabia with (i) Redis (synchronous) replication that is used widely in production, and (ii) RedisRaft \cite{RedisRaft}, which is Redis integrated with Raft. RedisRabia has comparable throughput with synchronous replication with the same fault-tolerance, and outperforms RedisRaft by 2.5x with batching.

\vspace{3pt}
\noindent\textbf{Challenges and Key Observation}:~
Randomized consensus algorithms are known to have a simpler structure than deterministic consensus algorithms \cite{Ben-Or_PODC83,Aspnes_RandomConsensus_Survey_DC2003}; however, there are two major challenges in deploying them in practical systems. First, most randomized algorithms have less stable performance due to randomized rules that depend on the outcome of a ``coin flip.'' That is, the performance is probabilistic in nature. Second, the latency of randomized algorithms is sub-optimal. All randomized consensus algorithms in the literature have at least four message delays \textit{on average}, whereas  prior deterministic consensus algorithms, e.g., \cite{lamport1998part,lamport2001paxos,Raft_ATC14}, have optimal fast-path latency of two message delays and stay on the fast path when there is a stable leader.

We observe that it is possible to make a randomized \emph{binary} consensus algorithm fast in current network infrastructures. 
In fact, \name{} achieves the best performance if a majority of replicas receive a similar set of messages, which is typical in modern networks. \name{} enables its underlying consensus algorithm to reach agreement in three message delays on the fast path. According to our tests using public commercial clouds and private clusters, current network infrastructures easily provide these characteristics (we call these networks \textit{stable}). 
This observation holds even with legacy network hardware (e.g., without recent RDMA cards). We also empirically show \name{} suffers only a minor performance loss when deployed in multiple availability zones in Google Cloud Platform. This implies that the network conditions that allow \name{} to reach consensus fast are practical, not stringent.

However, a binary consensus algorithm can only take a binary input ($0$ or $1$) and generate a binary output. The log-based SMR design requires a multi-valued consensus algorithm so that the replicas can agree on the order of client requests. Therefore, to make use of fast randomized binary consensus, the next technical challenge \name{} addresses is to efficiently convert binary consensus into a (weaker) form of multi-valued consensus that can be used for SMR.

\vspace{3pt}
\noindent\textbf{Key Techniques}:~
In \name{}, a client sends requests to an assigned replica, which then relays them to all the other replicas. Each replica stores pending requests that are not added to the log yet in its local \textit{min} priority queue (PQ). The key to the min PQ is the request timestamp. This design allows the replicas to have the same head of its PQ most of the time, even under high workload, in a stable network. Replicas then use our new consensus algorithm, \textsc{Weak-MVC}, to agree on the request for each slot of the log.

\textsc{Weak-MVC} is a novel implementation of a \textit{relaxed version} of multi-valued consensus. Each replica's input to Weak-MVC is the \textit{oldest pending request} in its local PQ.
The choice of using the oldest request as the input to our  algorithm allows us to achieve high performance in a stable network. 
If most replicas propose the same request, then \textsc{Weak-MVC} terminates in three message delays. In our evaluation, stable networks allow \name{} to use the fast path $99.58\%$ of the time under all open-loop experiments.

Replicas may propose different requests for the same slot. The approach to address this scenario is our second novel design -- \textit{forfeiting a slot} -- which is different from prior consensus algorithms, e.g., \cite{lamport1998part,lamport2001paxos,Raft_ATC14,EPaxos_SOSP13,Raynal_ISORC01,Raynal_SRDS04,Chen_IPL09}. Instead of deciding \emph{which} request to agree upon, we use binary consensus to determine \emph{whether} there is an agreement. If not, we choose to ``forfeit fast,'' i.e., allow a slot to store a NULL value. In this case, replicas can terminate (with a forfeited slot) also in three message delays. 

The rationale behind our design is that, in a stable network that delivers the oldest pending request $req_{old}$ to all replicas in a short time, replicas can soon reach an agreement on $req_{old}$ in a future slot, even if replicas forfeit current slot(s). Furthermore, a stable and reliable network also avoids continuous forfeits because eventually $req_{old}$ will be delivered. Compared to spending more communication to agree on a request, it is faster to forfeit the current slot, and move on to the next slot. Appendix \ref{app:paxos-noop} discusses how this approach differs from the
usage of a no-op operation in Multi-Paxos.

\vspace{3pt}
\noindent\textbf{Highlights of \name{}}:~
\name{} has the following features designed to favor widespread adoption and integration:

\begin{itemize}[nosep]

    \item 
    \textbf{No Fail-over and Simple Log Compaction}:~ One major challenge in implementing an SMR system is the integration of the auxiliary protocols, including fail-over, reconfiguration, snapshotting, and log compaction protocol (a garbage collection mechanism that discards older log slots from memory) \cite{Raft_MITBlog16,Paxos_live_PODC07,Raft_CockroachDB16,Jepsen_Raft_Redis,Jepsen_Raft_RethinkDB16}. Prior consensus algorithms, including both leader-based designs (e.g., Paxos \cite{lamport1998part,lamport2001paxos} and Raft \cite{Raft_ATC14}) and multi-leader-based designs (e.g., Mencius \cite{Mencius_Marzullo_OSDI08}, EPaxos \cite{EPaxos_SOSP13}, M$^2$Paxos \cite{M2Paxos_DSN16}, Caesar \cite{Caesar_DSN17}, and \textsc{Atlas} \cite{Atlas_Sutra_Eurosys20}), require a complicated fail-over mechanism to recover from the failure of a leader or a command leader, which also makes other auxiliary protocols challenging. 
    
    In contrast, \name{} does not need a fail-over. Intuitively, this is because \name{} uses randomized consensus instead of deterministic consensus to agree on the ordering of client requests, which ensures that at all times, non-faulty replicas are guaranteed to be able to learn the decision. This guarantee also enables a trivial log compaction mechanism in \name{} (which can be specified in three lines of pseudo-code), and simple reconfiguration and snapshotting protocols, as presented in Section \ref{s:discussion}.

    \item \textbf{Machine-checked proof}: In Section \ref{s:formal_verfication_safety}, we briefly describe how we use the Ivy~\cite{McMillanP20} and Coq~\cite{coq} theorem provers to formally verify the safety of our protocol. While a number of protocols and distributed systems have been verified, including Raft and many variants of Paxos~\cite{wilcox:verdi, GleissenthallKB19, TaubeLMPSSWW18, hawblitzel:ironfleet-cacm, PadonLSS17}, we are not aware of any verification of the multi-leader-based designs described above that are \name{}'s closest competitors. The full proof can be found at our GitHub repo \url{https://github.com/haochenpan/rabia/}. 

    \item \textbf{High performance}: 
    We implement \name{} in Go, and compare against Multi-Paxos and EPaxos in Google Cloud Platform, where the network is stable. In our setting, EPaxos with no conflicting workload is the best competitor. When deployed in the same availability zone with three replicas, \name's maximum speedup in throughput is up to 1.5x using a close-loop test, despite its quadratic communication complexity and sub-optimal latency on the fast path. In other cases, \name{} has comparable performance despite its higher message complexity.

\end{itemize}

\vspace{3pt}
\noindent\textbf{Outline}:~
We first enumerate the challenges of implementing existing consensus algorithms and discuss the background in Section \ref{s:prior-SMR-challange}. We then present the complete framework of \name{} and its analysis in Section \ref{s:rabia}. Practical considerations are discussed in Section \ref{s:discussion}. The proofs of our safety and liveness properties appear in Section \ref{s:correctness}. 
We evaluate \name{} in Section \ref{s:evaluation}. Related work is presented in  Section \ref{s:related}.
\section{Motivation and Background}
\label{s:prior-SMR-challange}

\begin{figure*}[t]
	\includegraphics[width=\linewidth]{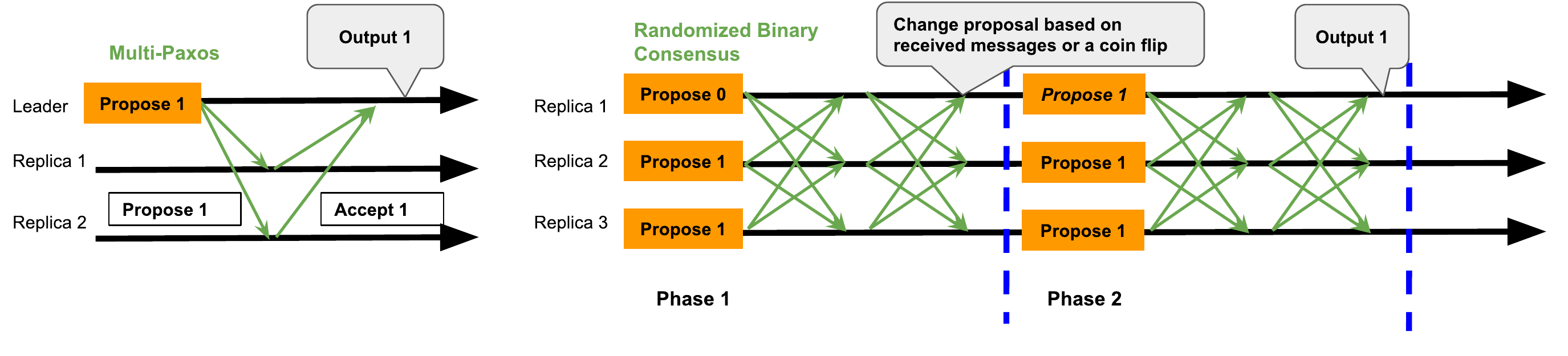}	
	\caption{\textbf{Multi-Paxos vs. Randomized Binary Consensus (Ben-Or's Algorithm \cite{Ben-Or_PODC83})}. Multi-Paxos relies on a leader to make a decision (for a slot), whereas Ben-Or's algorithm requires replicas to make a ``joint decision.'' In the common case, Multi-Paxos terminates in 2 message delays, whereas a randomized consensus algorithm might take multiple phases. In this particular example, Ben-Or's algorithm takes 2 phases (and 4 message delays) to terminate.}
	\label{fig:ben-or} 
	\vspace{-10pt}
\end{figure*}    

\vspace{3pt}
\noindent\textbf{Why Another Simple Consensus Algorithm?}~
The first barrier to implementing a consensus algorithm is understandability \cite{Raft_ATC14,Raft_CockroachDB16,PaxosComplex_ACMSurvey15,Paxos_live_PODC07}. The intuition behind Paxos and variants is difficult to grasp. Such a challenge was perfectly summarized by the developer of Google's Chubby system \cite{Paxos_live_PODC07} (which is based on Multi-Paxos): \textit{``There are significant gaps between the description of the Paxos algorithm and the needs of a real-world system $\dots$ the final system will be based on an unproven protocol.''} Most recent consensus algorithms proposed are inspired by Paxos, e.g., \cite{Mencius_OSDI08,EPaxos_SOSP13,M2Paxos_DSN16,Atlas_Sutra_Eurosys20}, and share a similar challenge in implementation.

Raft \cite{Raft_ATC14} addresses the understandability issue by using a stronger notion of leader to simplify the conceptual design. However, developing a Raft-based SMR system is still quite challenging, even for seasoned developers, e.g., \cite{Raft_CockroachDB16,Raft_MITBlog16,Jepsen_Raft_Redis,Jepsen_Raft_RethinkDB16}. It is because of the second barrier, namely \textit{engineering complexity}. 

While the authors of Raft provided a comprehensive overview of the necessary auxiliary protocols \cite{Raft_CockroachDB16,Raft_MITBlog16}, it is still difficult to develop and integrate them with the core consensus algorithm. For example, Redis \cite{Jepsen_Raft_Redis} and RethinkDB  \cite{RethinkDBRaft} had  bugs in their preliminary integration of Raft into their respective systems. In fact, even after Jepsen \cite{Jepsen_Raft_Redis} identified 21 major bugs back in June 2020, RedisRaft \cite{RedisRaft} showed several performance bugs in our experience (e.g., frequent leader changes when there is no failure, nor message drop). Industry blog posts \cite{Jepsen_Raft_Redis,Raft_CockroachDB16,Jepsen_Raft_RethinkDB16} identified fail-over, log compaction, reconfiguration, and snapshotting as major concerns when developing a Raft-based SMR system.

By using a randomized consensus as a core component, \name{} does not need a fail-over protocol, and hence supports trivial log compaction, snapshotting, and reconfiguration. In addition, \name{} uses two novel designs to achieve high performance in a stable network. We only found two competitors that achieve a similar goal: single-decree Paxos \cite{lamport1998part}\footnote{Single-decree Paxos is designed to reach an agreement on a single slot, whereas Multi-Paxos skips Phase 1 to achieve a higher performance with a stable leader in the case of multiple slots.} and NOPaxos (Network-Ordered Paxos) \cite{NOPaxos_OSDI16}. Section \ref{s:related} discusses recent Byzantine fault-tolerance systems (e.g., HoneyBadger \cite{HoneyBadger_CCS16} and Algorand \cite{Algorand_Micali_SOSP2017,Algorand_Micali_TCS2019}) that also utilize randomized consensus algorithms.

Single-decree Paxos does not need a fail-over, because any replica can help others recover an agreed value in any slot. However, it suffers from performance loss, compared to Multi-Paxos, and in general, has unstable performance. 
In fact, there is no bounded analysis on the average delay. \name{} is proven to have five message delays on average, as analyzed in Section \ref{s:correctness}. 
NOPaxos uses the network fabric to sequence requests (i.e., a sequencer in switches) to simplify design and improve performance. 
As a result, NOPaxos outperforms Multi-Paxos by 4.7x in throughput \cite{NOPaxos_OSDI16}. However, its fail-over and reconfiguration require switches to install new forwarding rules, which, in our opinion, make the development more complicated. In addition, most public clouds do not allow users to directly implement their own logic in switches, which makes NOPaxos less adoptable in general.

\label{s:pre}

\begin{figure*}[!ht]
	\includegraphics[width=\linewidth]{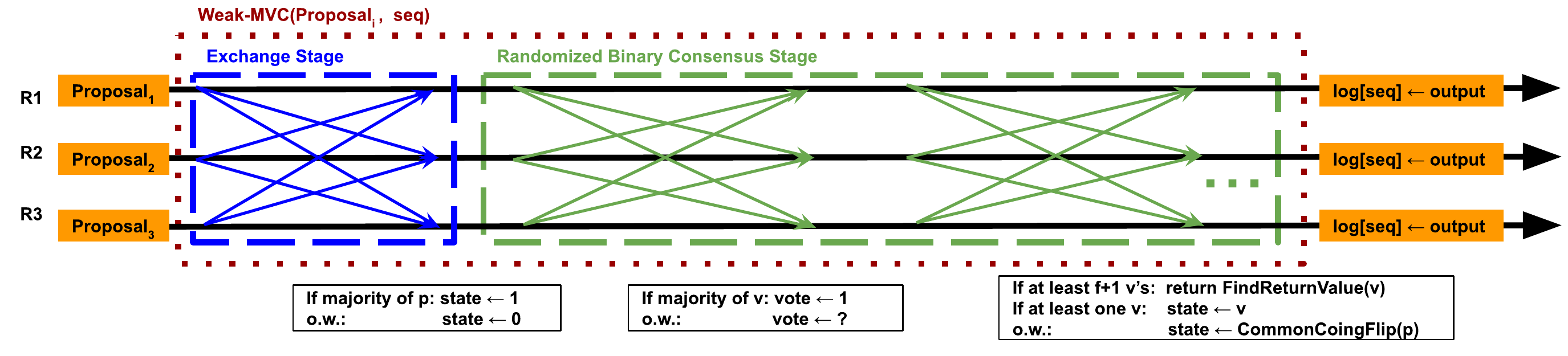}
	\vspace{-25pt}
	\caption{Execution Flow of \name{}. The core component is our consensus algorithm -- Weak-MVC -- whose communication pattern of the fast path is presented inside the red dotted box. The rules to update internal states are also listed underneath. }\label{fig:illustration} 
	\vspace{-10pt}
\end{figure*}

\vspace{3pt}
\noindent\textbf{SMR}:~
State-Machine Replication (SMR) ensures that a service is available and consistent in the presence of failures through replication. SMR ensures linearizability,  or strict serializability in case of transactions. 
We adopt the log-based SMR approach \cite{SMR_Schneider_ACM90,Mu_Aguilera_OSDI20,Raft_ATC14}, which works as follows: (i) each client submits a request to replicas (or servers) and waits for a response; (ii) each replica has a log that stores client requests; (iii) the log is divided into slots; (iv) replicas use a consensus protocol to agree on the ordering of the client requests, i.e., the request to be stored in each slot of the log; (v) replicas apply the requests (i.e., execute the operations in the request) slot by slot; and (vi) after applying the request, a replica sends a response to the corresponding client.

Following prior works (e.g., \cite{Mu_Aguilera_OSDI20,EPaxos_SOSP13,Raft_ATC14}), \name{} stores the log in main memory.
This design is reliable as long as the number of simultaneous replica failures is bounded. For deterministic services, all the replicas have the same state after applying the request in a given slot, because at that point, replicas have executed the same set of operations in the same order. 

\vspace{3pt}
\noindent\textbf{Consensus}:~
The core component of the log-based SMR design is its consensus algorithm (step (iv) above). A consensus protocol is correct if it ensures safety and liveness properties. For each slot of the log, liveness (or \textit{termination}) ensures that all the non-faulty replicas eventually obtain a request. Also, for each slot, safety states two conditions: \textit{agreement}: each non-faulty replica obtains the same request; and \textit{validity}: each obtained request was proposed by a client. 

The famous FLP result \cite{FLP} implies that deterministic consensus algorithms work only in partially synchronous systems. For example, Paxos (and its variants) and Raft ensure safety at all times, but ensures liveness only when the leader is stable and can reach a quorum of replicas.

\vspace{3pt}
\noindent\textbf{Randomized Consensus}:~
Another approach to circumvent the FLP result is through randomization. Ben-Or \cite{Ben-Or_PODC83} proposed a simple randomized binary consensus algorithm in PODC 1983, which achieves the following \textit{probabilistic termination} property: each non-faulty replica eventually obtains a request \textit{with probability 1} -- as time proceeds, the probability for a replica to terminate approaches 1. The reason for the probabilistic termination is that it breaks a tie by flipping a coin \cite{Ben-Or_PODC83}. Hence, an instance of this consensus protocol may take an arbitrarily long time to terminate. In fact, the average latency of Ben-Or's original algorithm is exponential with respect to the number of replicas, because each replica flips its own local coin. The latency is measured as the message delay from one replica to the other.

Figure \ref{fig:ben-or} compares the difference between Multi-Paxos and Ben-Or's algorithm. The left part illustrates Paxos (with a stable leader). 
Upon receiving an acknowledgment from another replica, the leader can safely output the value it proposed, value ``$1$'' in this example. In comparison,  Ben-Or's design does not rely on the notion of a leader. Instead, all the replicas proceed in phases, in which they need to propose and make a ``joint decision.''  Replicas may propose different values. In this case, replicas use a randomized rule (i.e., the outcome of a coin flip) to break a tie. In the example, Replica 1 changes from ``$0$'' to ``$1$'' in Phase 2. Then, all the replicas terminate and output $1$ in this phase.
In Section \ref{s:simple}, we discuss in more detail why such a ``joint decision'' design 
does not need a fail-over protocol, and hence simplifying the complexity in developing and integrating auxiliary protocols.

\vspace{3pt}
\noindent\textbf{Challenges}:~ 
Compared to highly optimized deterministic consensus algorithms like Paxos and Raft, randomized algorithms have a few limitations.  
As illustrated in Figure \ref{fig:ben-or}, the latency of a randomized consensus algorithm is less predictable, because the number of rounds may depend on the outcome of a coin flip. Moreover, the best-case latency is also worse than Multi-Paxos. Concretely, we identify two major challenges in designing a high-performance SMR system based on a randomized consensus algorithm: (i) how to stay on the fast path as often as possible? and (ii) how to reduce tail latency? We discuss how \name{} addresses these challenges in Section \ref{s:Rambia-challenge}.

\newcommand{\request}[0]{\textsc{Request}}
\newcommand{\proposal}[0]{\textsc{Proposal}}
\newcommand{\stateRabia}[0]{\textsc{State}}
\newcommand{\vote}[0]{\textsc{vote}}

\section{The \name{} SMR Framework}
\label{s:rabia}

We consider an asynchronous distributed system consisting of $n$ replicas (or servers). At most $f$ of the replicas may fail by crashing, i.e., a fail-stop model. \name{} ensures safety and liveness as long as $n \geq 2f+1$. We do not consider Byzantine faults. The set of replicas is assumed to be static. An approach for reconfiguring the set of replicas is presented in Section \ref{s:discussion}.

\subsection{Execution Flow}

\name{} adopts the log-based SMR design \cite{SMR_Schneider_ACM90}. Each replica has a log, and the goal is to agree on a client request (or a batch of requests) to be stored in each slot of the log. We present the pseudo-code of each replica $N_i$ in Algorithm \ref{algo:SMR-server}. Figure \ref{fig:illustration} illustrates the flow and rules for \name{} to agree on the request for a particular slot when $n=3$.

\begin{algorithm}[bth]
\begin{algorithmic}[1]
	\small
    \item[{\bf Local Variables}:]{}
    \item[] $PQ_i$ \COMMENT{priority queue, initially empty}
    \item[] $seq$ \COMMENT{current slot index, initially $0$}
    
    \item[] \hrulefill
    \item[{\bf Code for Replica} $N_i$:]{}
    \WHILE{$true$}
        \STATE $proposal_i \gets $ first element in $PQ_i$ that is not already in $log$ \\~~~ \COMMENT{$proposal_i = i$'s input to Weak-MVC}
        \STATE $output \gets$ \textsc{Weak-MVC($proposal_i, seq$)} 
        \STATE $log[seq] \gets output$ \COMMENT{Add $output$ to current slot}
        \IF{$output = \perp$ \textbf{or} $output \neq proposal_i$}
            \STATE $PQ_i.push(proposal_i)$ 
        \ENDIF
        \STATE $seq \gets seq+1$
    \ENDWHILE
    
    \vspace{3pt}
    
    \item[/* Event handler: executing in background */]
    
    \item[{\bf Upon receiving $\langle \request, c \rangle$ from client $c$}:]
    \STATE $PQ_i.push(\langle \request, c \rangle)$
    \STATE forward $\langle \request, c \rangle$ to all other replicas
    
    \vspace{3pt}
    
    \item[/* Executing in background periodically */]
    \item[\textbf{Log Compaction}:]

    \FOR{each $j$-th slot in the $log$}
        \IF{$log[j]$ has been executed locally}
            \STATE truncate $log[j]$ \COMMENT{Discard it or take a snapshot}
        \ENDIF
    \ENDFOR

\end{algorithmic}
\caption{\name: Code for Replica $N_i$}
\label{algo:SMR-server}
\end{algorithm}

A client sends a request to the assigned replica (a client proxy) and waits for a response from the same replica. Upon receiving the client request (Line 8 and 9), $N_i$ first pushes the request to a \textit{min priority queue} $PQ_i$ and then forwards it to all other replicas. $N_i$ uses $PQ_i$ to store pending client requests, i.e., requests that have not been stored in the log yet. The head of the priority queue stores the \textit{oldest pending request}. That is, $N_i$ uses the request timestamp as the key for $PQ_i$. For brevity, the timestamp is omitted in Algorithm \ref{algo:SMR-server}.

Replicas use a while loop to  continuously agree on the request for each slot $seq$ as long as their $PQ$ is non-empty (Line 1 to 7). All the replicas participate in our  consensus algorithm,  Weak-MVC (Weak Multi-Valued Consensus), to agree on a request for slot $seq$ (Line 2 and 3). $N_i$'s input to the Weak-MVC consists of a request extracted from $PQ_i$, and the current slot index $seq$. Due to concurrency and message delay, $PQ_i$ may contain some request $r$ that is actually already stored in the log, because replicas have already agreed to store $r$ in a prior slot. $N_i$ discards such requests, and moves to the next request in $PQ_i$. The condition at Line 2 can be checked efficiently using a dictionary. Section \ref{s:discussion} discusses why this dictionary will not grow unbounded. 

Once Weak-MVC terminates, the output is written into the current slot (Line 4). Weak-MVC is designed in a way that $N_i$ obtains an output with probability 1 as long as a majority of replicas, including $i$, are alive.
The output of Weak-MVC may be a proposal from some replica or a NULL value $\perp$.
A key design choice is to allow a slot to store a NULL value. In this case, $N_i$ learns that the slot is ``forfeited'' and has to push its proposal back into $PQ_i$ for a future slot (Line 6).
Figure \ref{fig:illustration} illustrates the flow of agreeing on a request for the slot $seq$. The red dotted box presents an example execution of Weak-MVC, whose design is presented in Section \ref{s:weak-MVC}.

Unlike prior SMR systems, e.g., \cite{Raft_ATC14,lamport1998part,lamport2001paxos,EPaxos_SOSP13}, performing log compaction in \name{} is very simple because the agreement and liveness properties of the Weak-MVC protocol ensure that all the non-faulty replicas will eventually agree on the same request or NULL for each slot. Therefore, each replica can discard slot $j$ or take a snapshot once it has executed the request stored in the slot (Line 10 to 12). In prior systems, the log compaction mechanism needs to be carefully integrated with a leader election or a fail-over protocol to ensure that the discarded log can be safely recovered if failures occur. Such integration is not always straightforward \cite{Paxos_live_PODC07,Atlas_Sutra_Eurosys20,Raft_ATC14}.

\subsection{Addressing performance challenges of randomized consensus}
\label{s:Rambia-challenge}

Before delving into details of Weak-MVC, let us first investigate how our framework presented in Algorithm \ref{algo:SMR-server} addresses the challenges mentioned in Section \ref{s:pre}: (i) staying on the fast path as often as possible and (ii) avoiding long tail latency in a stable network. 
As it will become clear later, Weak-MVC is guaranteed to use the fast path (which \textit{terminates in three message delays}), if 
\begin{itemize}[nosep]
    \item (i) all replicas have the same proposal; or
    
    \item (ii) no majority of replicas propose the same request.
\end{itemize}
While these two conditions might not always occur, in our experience, they were often satisfied using a stable network.

Condition (i) together with using the \textit{oldest pending request} (stored as the head of $PQ_i$) as a proposal allows  the agreement to be made using mostly the fast path in a stable network. 
This is based on the two following observations: first, each replica forwards a request to other replicas when it receives an incoming request (Line 9 of Algorithm \ref{algo:SMR-server}); and second, if message delay is small compared to the interval between two consecutive requests, then it is highly likely that most replicas have the same oldest pending request $r$ in its local priority queue shortly after that request $r$ was sent.

Condition (ii) reduces the chance for a long tail latency. In most randomized binary consensus algorithms (e.g., \cite{Raynal_ISORC01,Raynal_SRDS04,Ben-Or_PODC83,ByzRandConsensus_Raynal_PODC14}), the slow case occurs when roughly half of the replicas propose $0$ and the others propose $1$. In an unlucky sequence of message delays and coin flips, it might take a long time to terminate because replicas do not observe enough concordant values to jointly make a decision. Therefore, we design Weak-MVC in a way that if it seems difficult to terminate fast, then replicas choose to forfeit the proposal, and Weak-MVC outputs a NULL value $\perp$ in this case. 

In our experience, it is much more efficient to ``forfeit a slot'' than to agree on some client request.  
Recall that the cause behind a forfeit is that replicas had extracted different proposals. During the time that replicas forfeit a slot, the oldest proposal that has not been agreed upon is highly likely to be propagated to all the replicas in a stable network. These replicas would then store this proposal in local PQs, and \name{} can hit the fast path again on the next slot after forfeiting. Therefore, our forfeit-fast approach is more efficient than agreeing on a non-NULL value at all times. The latter approach \cite{Raynal_ISORC01,Raynal_SRDS04,Chen_IPL09} takes logarithmic time with respect to the size of the client requests. 
Section \ref{s:discussion} presents two other practical solutions that further improve the fast-path latency and reduce tail latency.

\subsection{Weak-MVC}
\label{s:weak-MVC}

Allowing a slot to forfeit and contain a NULL value, \name{} violates the validity property by typical consensus algorithms (i.e., an output has to be a client request). Instead, it achieves the following relaxed version of validity:

{
\setlength{\fboxsep}{0pt}%
 \setlength{\fboxrule}{0pt}
\begin{center}
\framebox{%

  \begin{minipage}{0.8\linewidth}
    \textbf{Weak Validity}: the value stored in each slot of the log must either be a request from some client or a NULL value $\perp$.
  \end{minipage}
  }
\end{center}  
}

Our algorithm, Weak-MVC (Weak Multi-Valued Consensus), achieves agreement, weak validity, and probabilistic termination as long as a majority of replicas are correct. The pseudo-code is presented in Algorithm \ref{algo:Weak-MVC} and \ref{algo:Weak-MVC-helper}.

\vspace{3pt}
\noindent\textbf{Pseudo-code and Illustration}:~ Weak-MVC has two stages: exchange stage (Line 1 to 7  in Algorithm \ref{algo:Weak-MVC}) and randomized binary consensus stage (Line 8 to 27  in Algorithm \ref{algo:Weak-MVC}). 
The first stage takes one phase (one message delay).  The second stage takes one phase (two message delays) on the fast path and may take more phases in an unlucky sequence of message delays and coin flips.

\begin{algorithm}[tb!]
\begin{algorithmic}[1]
	\small

    \item[{\bf When \textsc{Weak-MVC} is invoked with input $q$ and $seq$:}]
    \STATE \slash \slash~ Exchange Stage: exchange proposals
    \STATE Send ($\proposal, q$) to all  \COMMENT{$q$ is client request}
    \STATE \textbf{wait until} receiving $\geq n-f$~~ \proposal~ messages
    \IF{request $q$ appears $\geq \lfloor\frac{n}{2}\rfloor + 1$ times in \proposal s}
        \STATE $state \gets 1$ 
    \ELSE
        \STATE $state \gets 0$ 
    \ENDIF
    
    \vspace{3pt}
    
    \STATE \slash \slash~ Randomized Binary Consensus Stage (Phase $p \geq 1$)  
    \STATE $p\gets 1$ \COMMENT{Start with Phase $1$}
    \WHILE{$true$}
        \STATE /* Round 1 */
        \STATE Send ($\stateRabia, p, state$) to all \COMMENT{$state$ can be $0$ or $1$}
        \STATE \textbf{wait until} receiving $\geq n-f$ phase-$p$ \stateRabia~messages
        \IF{value $v$ appears $\geq \lfloor\frac{n}{2}\rfloor + 1$ times in \stateRabia s}
            \STATE $vote \gets v$
        \ELSE
            \STATE $vote \gets ?$
        \ENDIF
        \STATE /* Round 2 */
        \STATE Send ($\vote, p, vote$) to all \COMMENT{$vote$ can be $0, 1$ or $?$}
        \STATE \textbf{wait until} receiving $\geq n-f$ phase-$p$~~ \vote~ messages
        \IF{a non-$?$ value $v$ appears $\geq f+1$ times in \vote s}\label{line:majv>=f+1}
            \STATE \textbf{Return} \textsc{FindReturnValue}($v$) \COMMENT{Termination}
        \ELSIF{a non-$?$ value $v$ appears at least once in \vote s}
            \STATE $state \gets v$ \label{line:maj-value}
        \ELSE
            \STATE $state \gets \textsc{CommonCoinFlip}(p)$ \COMMENT{$p$-th coin flip} \label{line:coin-value}
        \ENDIF
        \STATE $p \gets p+1$ \COMMENT{Proceed to next phase}
    \ENDWHILE
\end{algorithmic}
\caption{Weak-MVC: Code for Replica $i$}
\label{algo:Weak-MVC}
\end{algorithm}

\begin{algorithm}[tb!]
\begin{algorithmic}[1]
	\small

    \item[\textbf{Procedure} \textsc{FindReturnValue}($v$)]
    \IF{$v = 1$}
        \STATE Find value $m$ that appears $\geq \lfloor \frac{n}{2} \rfloor + 1$ times \\~~~~~~~~~~
        in \proposal~messages received in Phase $0$
        \STATE \textbf{Return} $m$
    \ELSE    
        \STATE \textbf{Return} $\perp$ \COMMENT{return null value}
    \ENDIF    
    
\end{algorithmic}
\caption{Weak-MVC: Helper Function}
\label{algo:Weak-MVC-helper}
\end{algorithm}

The second stage is an adaption of Ben-Or's algorithm \cite{Ben-Or_PODC83}. Intuitively, we replace the per-replica local coins used by Ben-Or with a common coin (Line 26 of Algorithm \ref{algo:Weak-MVC}), whose value is the same across all replicas in the same phase. Section \ref{s:discussion} briefly discusses an efficient implementation of a common coin by having all replicas use a pseudo-random number generator or a hash with a common seed \cite{Fault-tolerantBook_Raynal_Book18}. The seed can be configured at deployment, so there is no communication involved when using the common coin.

The red dotted box in Figure \ref{fig:illustration} presents the flow and the rules of Weak-MVC. The exchange stage and randomized binary consensus stage are shown in the blue and green dotted box, respectively. 
For brevity, we present the case when the randomized binary consensus stage terminates in one phase. The three green dots at the end of the red box denote the possibility that this stage might need more phases to agree on an output. We will prove later in Section \ref{s:correctness} that Weak-MVC terminates with probability 1, and its average latency is five message delays.

\vspace{3pt}
\noindent\textbf{Algorithm Description}:~
Weak-MVC uses three types of messages: (i) \proposal{} message that carries a proposed value in the exchange stage (the input to  Weak-MVC); (ii) \stateRabia{} message that contains the state value, denoted by the variable $state$ in the pseudo-code; and (iii) \vote{} message that is used to cast votes in each phase. The value in the \proposal{} message is a client request. The state value is either $0$ or $1$, whereas the vote could take the value of $0, 1$ or $?$. The $?$ is a unique symbol denoting that the replica does \textit{not} vote for $0$ or $1$, i.e., it gives up the vote in this phase. 
For messages sent in phase $p$, we will refer to it as phase-$p$ messages. 
Each message is tagged with slot index $seq$, phase index $p$, and sender ID. For brevity, we omit $seq$ and ID in the pseudo-code.

In the exchange stage, replicas exchange their proposals (a client request) at Line 2 of Algorithm \ref{algo:Weak-MVC}, and decide the input to the second stage by updating the $state$ variable based on the received proposals. In all communication steps, replicas wait for $n-f$ messages of a certain type (Line 3, 13, and 20 of Algorithm \ref{algo:Weak-MVC}). Thus, Weak-MVC is \textit{non-blocking}. Each replica sets $state$ to $1$ if it sees a majority of messages with the same request $q$ 
(Line 5); otherwise, it sets $state$ to $0$ (Line 7). Using the majority as the threshold ensures that two replicas with $state = 1$ at the end of the exchange stage must have seen the same proposal from a majority of replicas. 

\begin{figure*}[!ht]
	\includegraphics[width=\linewidth]{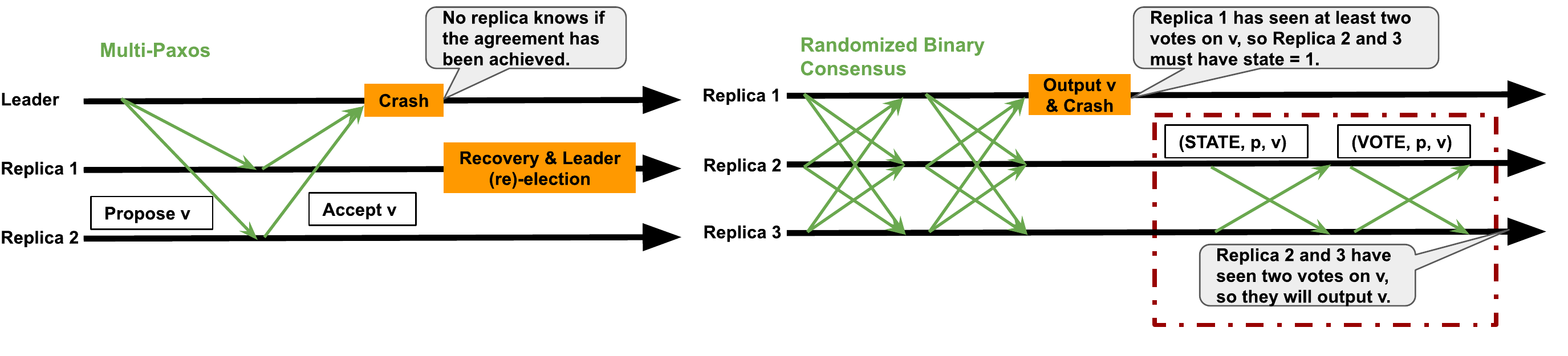}	
	\vspace{-25pt}
	\caption{\textbf{Fail-over in Multi-Paxos vs. No fail-over in \name{}.} If the leader in Multi-Paxos fails, then the remaining replicas need to execute a fail-over protocol to recover decisions made by the crashed leader and re-elect a new leader. \name{} does not need any fail-over. If Replica 1 outputs a value $v$, then the remaining replicas are guaranteed to output $v$ in the next phase, no matter if Replica 1 fails or not. For brevity, only the randomized binary consensus stage of Weak-MVC is shown. }\label{fig:failover} 
	\vspace{-10pt}
\end{figure*}

The second stage, the randomized binary consensus stage, proceeds in phases, and each phase has two rounds. In the first round of phase $p$, replicas exchange their $state$ variables and decide their vote for this phase. Each replica sets $vote$ to $v$ if they have seen a majority of phase-$p$ $\stateRabia$ messages containing $v$, which is either $0$ or $1$; otherwise, the replica updates $vote$ to $?$. 
Using a typical quorum intersection argument on $state$ messages, it is guaranteed that there is at most one non-$?$ value in all the votes. That is, in the same round, we have either $vote \in \{0, ?\}$ or  $vote \in \{1, ?\}$ at all the replicas. The proof of the safety property relies on this observation.

In the second round of phase $p$, replicas exchange their votes stored in the $vote$ variable. If a non-$?$ value $v$ appears at least $f+1$ times in the received $\vote$ messages, then replicas can safely output a value using the helper function (Line 22). If a non-$?$ value $v$ appears at least once, then the replica sets $state$ for the next phase as $v$ (Line 24). If all phase-$p$ $\vote$ messages received contain $?$, then the replica flips a \textit{common coin} to determine the state for the next phase (Line 26) -- this is the randomized rule in Weak-MVC. \textsc{CommonCoinFlip(p)} denotes the $p$-th coin flip at each replica. A common coin guarantees that the $p$-th coin flip at each replica returns an identical value. That is, all the replicas that execute Line 24 in the same phase must obtain the same $state$ variable. Such a coin can be easily implemented in our scenario as detailed later in Section \ref{s:discussion}. Replicas proceed to the next phase if they have not already terminated at Line 22.

The thresholds used ensure that all the non-$?$ values contained in phase-$p$ $\vote$ messages are identical. Hence, if replicas execute Line 22 or Line 24, they will see the same $v$. The beauty of Ben-Or's design, which we borrowed, is that replicas break a tie using a coin flip. Recall that a non-$?$ value $v$ is either $0$ or $1$. Hence, when replicas flip a common coin at Line 26, it has probability $1/2$ to flip to $v$ and will terminate in the next phase. Section \ref{s:correctness} 
extends this observation to obtain the probabilistic termination property.

\vspace{3pt}
\noindent\textbf{Examples}:~
The fast path is three message delays, as illustrated in Figure \ref{fig:illustration}. 
Section \ref{s:Rambia-challenge} mentioned two cases that Weak-MVC is guaranteed to use the fast path. We now present the details.
Consider the case when each replica has the same proposal, then Weak-MVC  terminates in three message delays, because (i) all replicas have $state = 1$ after the exchange stage; (ii) all replicas have $vote = 1$ after round 1 of phase $1$; and (iii) all replicas will execute Line 22 because all the $\vote$ messages contain $1$. Another case of the fast path is when each replica has a unique proposal. In this case, (i) all replicas have $state = 0$ after the exchange stage; (ii) all replicas have $vote = 0$ after round 1 of phase $1$; and (iii) all replicas will execute Line 22 because all the $\vote$ messages contain $0$.

Weak-MVC might be slow if in every communication step, only roughly half of replicas take a particular branch of the if statements, particularly when half of the replicas execute Line 24 and the other half execute Line 26. A slow-case example with $n = 3$ is as follows: (i) two replicas have $state = 1$ and one has $state = 0$ at the beginning of round 1 in the second stage; (ii) two replicas have $vote = ~?$ and one has $vote = 1$ at the beginning of the round 2, because they observed different sets of $\stateRabia$ messages; and (iii) two replicas execute Line 24 with $state = 1$, and one replica executes Line 26 by flipping a coin and obtains $0$ for its $state$. Note that step (iii) is identical to step (i). In other words, no progress has been made. This sequence of events might repeat for a few phases if the coin flip is different from $v$ at Line 26. 
Section \ref{s:correctness} proves that the average latency of \name{} is five message delays. Our evaluation in Section \ref{s:evaluation} shows that such a long tail latency rarely occurs in a stable network.

\subsection{No Fail-Over and Simple Log Compaction}
\label{s:simple}

Prior systems, e.g., \cite{lamport1998part,lamport2001paxos,Raft_ATC14,EPaxos_SOSP13,Atlas_Sutra_Eurosys20,Mencius_Marzullo_OSDI08}, require a complicated protocol to handle failures, especially when a leader, or a command leader in a multi-leader design, fails. Intuitively, this is because when a leader crashes, replicas need to execute an instance of leader election to ensure liveness. Meanwhile, for safety, replicas need to recover the decision(s) that the previous leader has already made \cite{Atlas_Sutra_Eurosys20,Paxos_live_PODC07,Raft_CockroachDB16}. This task is non-trivial since replicas may have a stale log and may fail again during leader election. In practical systems that optimize throughput with pipelining and leases (for leader or read operations), fail-over becomes even more complicated, because some of the decisions may not have been persisted to a quorum of replicas yet. An example scenario of a crashed leader in Multi-Paxos is depicted in Figure \ref{fig:failover}.

\name{} does not need any fail-over protocol because it does not rely on a notion of leader, nor command leader. To agree on a value $v$ for a slot, a group of replicas must cast votes for $v$ in some phase of the randomized binary consensus stage; hence, even if some replica(s) fails, a replica can still learn $v$ from the non-faulty replicas in the group. In other words, the use of the $\vote$ messages ensures that no single replica makes the decision alone, unlike prior algorithms. In \name, a group of replicas makes a ``joint decision.'' This design simplifies the failure handling. 

The execution on the right part of Figure \ref{fig:failover} illustrates the case when there are three replicas, and only Replica 1 observes two $\vote$ messages of $v$ and executes Line 22 of Algorithm \ref{algo:Weak-MVC}. Then, it fails shortly after it learns the output. This scenario is analogous to the case when the only replica that knew the output fails in Multi-Paxos. 
Suppose Replica 2 and 3 did not learn the decision in Phase 1 and proceed to the next phase as described in Algorithm \ref{algo:Weak-MVC} with $state = v$. This is because they must have observed at least one $\vote$ message of $v$ in Phase 1. By construction, after the exchange of phase-2 $\stateRabia$ messages, both Replica 2 and 3 update $vote$ to $v$; hence, they will output $v$ in Phase 2. Note that these steps are already specified in Algorithm \ref{algo:Weak-MVC}, and Weak-MVC does \textit{not} need an auxiliary protocol to handle failures.

Log compaction is difficult in prior systems because a leader needs to send a snapshot to replicas that lag behind \cite{CockroachDBRaft,Paxos_live_PODC07}. The issue becomes even more challenging during the process of leader re-election and recovery.  In fact, Chubby developers \cite{Paxos_live_PODC07} mentioned that ``\textit{(Snapshotting) mechanism appears straightforward at first... However, it introduces a fair amount of complexity into the system.}'' Particularly, the meta-data needs to be stored along with the snapshot itself, which both need to be recovered consistently if a leader failure occurs. \name{} does not need a fail-over and does not rely on the notion of leader; hence, it supports a simple mechanism for log compaction. In particular, each replica has the same responsibility in \name{}, so a slow replica can learn from any other replica to catch up with missing slots in a stable network where no message is lost between non-faulty replicas. Once a slow replica gets to the current slot, it can participate in the agreement process again. If messages could be lost, then log compaction needs to be made more conservative than the one presented in Algorithm \ref{algo:SMR-server} -- replicas need to make sure that a quorum has already agreed on a slot before compacting a slot. 

\subsection{Performance Analysis and Trade-off} 

Rabia does not need a fail-over protocol, but it demands more communication, which we believe is a good trade-off in the target scenario. As we will see in Section \ref{s:evaluation}, \name's performance with three replicas in the same availability zone is better than prior systems. Although we observe a degrading performance with a larger $n$ or deployment across different availability zones (in a datacenter), \name{} still has performance comparable with EPaxos.

\vspace{3pt}
\noindent\textbf{Fast Path Case}: For $n = 3$, we analyze the best-case performance. Both Multi-Paxos and EPaxos (with zero conflict rate) need two message delays, whereas \name{} has three message delays. Mutli-Paxos sends four messages, and EPaxos sends four messages per command-leader.  
\name{} sends 18 messages. Note that Weak-MVC only sends the original proposal (a request) in the exchange stage and sends a message containing $0, 1$ or $?$ in later phases. That is, bit complexity is dominated by client request size. Therefore, \name{}'s communication overhead is not a bottleneck with small $n$ in a stable network.

\vspace{3pt}
\noindent\textbf{Average and Worst Case}:
We prove later in Section \ref{s:correctness} that the average latency of \name{} is five message delays on average. 
As stated earlier, it is possible, in theory, to have a large tail latency due to the slow case of Weak-MVC; however, we argue that such a corner case hardly occurs in practice, as witnessed in our evaluation. In Section \ref{s:discussion}, we present two practical approaches to mitigate such an unlucky scenario.

\vspace{3pt}
\noindent\textbf{Bottleneck}:
In our evaluation, we found that the bottleneck of Multi-Paxos is the leader. It requires a stable leader to maintain high-performance; however, leader-client and leader-replica communications, especially serialization/deserialization,  become the bottleneck. We see 1.5x to 3x difference in throughput between open-loop tests and close-loop tests. As for EPaxos, the bottleneck is local computation time, i.e., dependency check for identifying a permissible order across client requests. Even in the case when there is no conflict, EPaxos still needs to check all the dependencies to ensure safety. The check is proportional to the number of clients, replicas, and the number of client requests in a batch. With 100 clients and batch size 1, the median time for EPaxos to complete the dependency check is $0.29$ms, roughly one RTT in the current datacenter. The time increases to $1.90$s when batch size increases to 80. Appendix \ref{app:EPaxos-local-time}  presents its median time for key functions with 100 clients and varying batch size.

There are two characteristics of \name{} that can potentially reduce performance: message complexity and stable network. We found the trade-off reasonable. 
\name{} does not need an auxiliary fail-over protocol; moreover, it distributes workload evenly to all replicas, and does not have expensive local computation. As we will see in Section \ref{s:evaluation}, these are why \name{} outperforms Paxos and EPaxos when $n=3$ in the same availability zone.

\section{Practical Considerations}
\label{s:discussion}

We implemented \name{} in Go, version 1.15.\footnote{\url{https://github.com/haochenpan/rabia/}} The entire framework, including the key-value storage on top of our system, consists of around 2,200 lines of Go code. Go is a compiled language and has garbage collection and built-in support for managing high concurrency. The implementation follows closely with the pseudo-code. We use go-routines and go-channels for handling connection and communication. This simplifies the logic and increases parallelism. 

\vspace{3pt}
\noindent\textbf{Common Coin}:~ \name{} uses a common coin (Line 26 in Algorithm \ref{algo:Weak-MVC}). We implement it by using a random binary number generator with the same seed across all replicas. This satisfies our purpose by (i) picking $0$ or $1$ randomly in each phase; and (ii) giving the same value to all replicas in the same phase, i.e., all replicas have the same $p$-th coin flip. A seed is pre-configured in such a way that for each slot, all replicas use a common seed. When the system is reconfigured, the seeds are reset using a deterministic rule, i.e., the slot index plus the configuration index (epoch number) decide the seed.

The implementation of our common coin is based on the seminal work by Rabin \cite{Rabin_FOCS83_randomizedByz} and subsequent works by Michel Raynal \cite{Raynal_SRDS04,Raynal_TDSC05_randomizedByz} on binary randomized Byzantine consensus. 
Ben-Or's local coin approach \cite{Ben-Or_PODC83} tolerates a dynamic adversary, which may determine the set of faulty replicas and delay messages based on the coin flip. The use of a common coin tolerates a weaker adversary which assumes that the faults and message delays are independent of the outcome of the coin flip, which is adequate in  practical settings. 

\vspace{3pt}
\noindent\textbf{Dictionary}:~
At Line 1 of Algorithm \ref{algo:SMR-server}, \name{} uses a dictionary to keep track of whether the head of the $PQ$ is already in the log or not.  \name{} only needs to store a non-$\perp$ $output$ (the output of \textsc{Weak-MVC}) in the dictionary at Line 6, when this $output$ does not match $proposal_i$. TCP implies that the communication is reliable, and a message is delivered exactly once, if the sender is correct. Therefore, replica $i$ is guaranteed to extract $output$ again from its $PQ_i$ when there is no failure. 
At this step, $output$ can be removed from the dictionary. By assumption, up to $f$ replicas may fail, so $PQ_i$ still has a bounded size if some replicas crash. In our experience, the size of the dictionary is mostly empty in a stable network.

\vspace{3pt}
\noindent\textbf{Batching}:~
Following prior work, e.g., \cite{Mencius_Marzullo_OSDI08,EPaxos_SOSP13,M2Paxos_DSN16,Caesar_DSN17}, we use batching to improve throughput and network utilization. That is, instead of agreeing on one client request for a slot, replicas can agree on a batch of requests. We implement two forms of batching:
\begin{itemize}[nosep]
    \item \textit{Proxy batching}: replicas batch a fixed number of client requests before pushing it into the priority queue $PQ$ (Line 8 of Algorithm \ref{algo:SMR-server}), and forwarding to other replicas (Line 9). Consequently, an entire batch is treated as a proposal at Line 2.  If there is no failure, then each batch contains different requests.

    \item \textit{Client batching}: instead of sending a single request, clients send a batch of requests (i.e., a set of operations) in one message. In practical systems, user applications typically communicate with a load balancing service, which can perform this type of batching by collecting requests from multiple clients \cite{NetCache_Stoica_SOSP17,fitzpatrick2004distributed}. Moreover, as outlined in the experience paper of scaling Memcache by Facebook \cite{Facebook_Memcache_NSDI13}, a query touches 24 keys on average. Such a query can be viewed as a batch of 24 requests to the key-value storage system.
\end{itemize}

\vspace{3pt}
\noindent\textbf{Failure Recovery by Clients}:~
In our design, a client communicates with a single replica, i.e., a proxy replica. If the proxy replica fails, then the client relies on a timeout to detect the unresponsiveness. In such a case, the client re-sends its request to another (randomly selected) replica. To handle the potential duplicated requests in the log, we use the standard solution of embedding a unique ID (a pair of client ID and a sequence number) in each request. When applying requests from the log, a replica skips any duplicated request. Appendix \ref{app:failure} presents an experiment demonstrating that Rabia has a quick recovery from a replica failure.

\vspace{3pt}
\noindent\textbf{Reconfiguring the Replicas}:~ Since \name{} does not need a fail-over, and all replicas essentially have equal responsibility, it is simple to perform reconfiguration. We treat \texttt{add-replica} and \texttt{remove-replica} as special commands. A system administrator (or an auxiliary automated membership management component) can submit a special command $c$ to any of the replicas. Replicas will then use \textsc{Weak-MVC} to agree on the slot for $c$. Eventually, all replicas will learn $c$, and in the next slot, if $c$ is \texttt{add-replica}, then the new replica will join the protocol; otherwise, the removed replica will leave the system. The reason that the reconfiguration is simple compared to prior systems \cite{Raft_ATC14,EPaxos_SOSP13,VerticalPaxos_Lamport_PODC09} is that \name{} does not rely on the notion of the leader; therefore, every replica is eventually obtaining the same information, which allows them to change to a new configuration jointly. Prior systems need to carefully integrate reconfiguration with leader election to deal with potential leader failures \cite{Raft_CockroachDB16,Paxos_live_PODC07}.

\vspace{3pt}
\noindent\textbf{Tail Latency Reduction}:~We briefly describe two practical approaches to further reduce tail latency: (i) Using an eventually correct failure detector \cite{phi_failureDetector_SRDS04} to allow replicas to receive a consistent set of messages; and (ii) using a freeze time before participating in Weak-MVC, which allows the oldest pending request to be delivered at replicas. The rationale behind these two approaches is to increase the probability of having the same head of each $PQ_i$ for replica $N_i$, which allows \name{} to take the fast path.\footnote{Appendix \ref{app:tail-latency-reduction} presents more detail.} However, these two are optimizations that are \textit{not} currently implemented because the network used in our experiments, Google Cloud Platform, shows stability, and hence none of these optimizations should have had any impact on our performance study. 

In addition, we adopt a common optimization that allows a slow replica to catch up by asking other replicas when it misses the proposal from a prior slot. That is, a slow replica that learns the decision of a slot may send a request message to other replicas to learn the proposal that has already been agreed upon. This allows the slow replica to participate in the next slot without waiting for delayed messages.

\vspace{3pt}
\noindent\textbf{Pipelining}:~ 
Pipelining is a common optimization that allows the system to proceed to the next slot(s) without learning the output of the current slot. In other words, multiple instances of consensus algorithms are being executed simultaneously. This optimization increases throughput because communication latency is amortized through concurrent slots. 

While \name’s implementation does not include pipelining, our framework can be extended to support it. As presented in Algorithm \ref{algo:SMR-server}, each replica has one PQ for providing inputs to Weak-MVC. To enable pipelining, we can have multiple PQs for each replica. Then \name{} has one PQ to handle the request batches from a fixed set of replicas, and multiple instances of Weak-MVC can run concurrently and independently. Since randomization ensures that each instance is guaranteed to terminate, liveness still holds with this pipelining strategy.

\section{Safety and Liveness Properties}
\label{s:correctness}

The design of Rabia lends itself to formal proofs of correctness.  We briefly discuss how we use Ivy and Coq to formally verify its safety.\footnote{ 
The complete proof can be found at our GitHub repo \url{https://github.com/haochenpan/rabia/}.} 
We then present a simple analysis of liveness.

\vspace{3pt}
\noindent\textbf{Formal Proof of Safety}:~
\label{s:formal_verfication_safety}
The safety of Weak-MVC implies the safety of \name. Weak-MVC's key safety properties are \emph{weak validity} and \emph{agreement}. In the discussion below, we say that a replica ``decides'' on a non-$?$ value $v$ in phase $p$ if the replica executes Line 22 of Algorithm \ref{algo:Weak-MVC} after seeing a majority of phase-$p$ $\vote$ messages with value $v$.
By construction, replicas will decide on only $0$ or $1$.

Weak validity says that if a replica decides on a value $v$ other than $\perp$, some client must have initially proposed $v$. Agreement says that if a replica decides on $v_1$ and another replica decides on $v_2$, then $v_1 = v_2$. These are the usual properties expected of a correct consensus algorithm, except that weak-validity allows for replicas to decide on $\perp$.

We prove these two safety properties using a combination of the Ivy~\cite{McMillanP20} and Coq~\cite{coq} verification tools. Ivy is an automated tool that checks whether a property is an \emph{inductive invariant} of a system, meaning that it holds before and after every action of a replica. Automated proof search is achieved by careful restrictions on the types of properties that can be expressed and checked in order to ensure that they fall within a decidable fragment of first-order logic. Meanwhile, Coq is a general-purpose theorem prover, which is highly expressive but requires a human to construct the proof.

The core part of our safety proof is showing that the randomized binary consensus stage of Weak-MVC (Algorithm \ref{algo:Weak-MVC}) satisfies agreement. For this, we follow the structure of prior paper proofs of correctness for Ben-Or's algorithm~\cite{AguileraT12}. 
We first use Ivy to verify that the following four properties are inductive invariants of the system:
\begin{enumerate}[nosep]
 \item Any two decisions within a phase must be on the same value $v$.
 \item Once a replica decides on a value $v$, the next phase is \emph{value-locked} on $v$, meaning that all the replicas that have neither crashed nor decided must enter the next phase with $state=v$.
 \item If a phase is value-locked on $v$, any decisions within that phase must be for $v$.
 \item If phase $i$ is value-locked on $v$, then $i+1$ is also value-locked on $v$. 
 \end{enumerate}
 We next use Coq to prove that agreement follows from the above four properties by induction on the phase number. 
 
 Splitting the proof across the two tools in this way makes use of the relative strengths of each. Checking the four properties above is straightforward in Ivy but would require many lines of proof in Coq. Conversely, trying to do the induction on the phase number within Ivy seems difficult to do while remaining in the supported decidable fragment.\footnote{Recent versions of Ivy have some support for including interactive proofs, which might allow doing this proof entirely within Ivy.}
 
 The Ivy part of the proof is 366 lines of code, about 150 of which are a description of Weak-MVC in Ivy's modeling language. The remainder are statements of other intermediate inductive invariants that Ivy uses to establish the four properties above. The Coq file is 190 lines of code, with 128 lines describing the system, axiomatizing the properties checked by Ivy, and stating the agreement and weak validity.

\vspace{3pt}
\noindent\textbf{Proof of Liveness}:~
For liveness, we need to show that Weak-MVC terminates with probability $1$, as long as the majority of replicas are non-faulty. The proof structure follows from prior works \cite{Ben-Or_PODC83,Aspnes_RandomConsensus_Survey_DC2003,Fault-tolerantBook_Raynal_Book18,ByzRandConsensus_Raynal_PODC14}. We first prove the following lemma and use it to prove three key theorems. 

\begin{lemma}
\label{lemma:termination-probability}
Each phase has probability at least $\frac{1}{2}$ of leading to termination. That is, all the non-faulty replicas output a value before or at the end of phase $p$.
\end{lemma}

\begin{theorem}
\label{thm:expected-constant-time}
Weak-MVC has average round complexity $=5$.
\end{theorem}

\begin{proof}
Lemma \ref{lemma:termination-probability} implies that the probability of Weak-MVC's Randomized Binary Consensus Stage terminating at the end of phase $t$ is at least

\begin{equation}
\label{eq:bound}
(1-\frac{1}{2})^{t-1} \frac{1}{2}
\end{equation}

It follows that the number of phases until termination is upper bounded by a geometric random variable whose expected value is $1/2$. Therefore, the average number of phases of this stage is $2$.
Since each phase has two rounds and Weak-MVC also needs one round in the Exchange Stage, the average number of rounds is $2 \cdot 2 + 1 = 5$. 
\end{proof}

Equation (\ref{eq:bound}) implies that Weak-MVC terminates with probability 1. Since \name{} only performs communication in Weak-MVC, the average number of rounds for \name{} to complete a slot is 5, and each slot is completed with probability 1.

\begin{figure*}[t]
\begin{subfigure}{.24\textwidth}
  \centering
  \includegraphics[width=\textwidth]{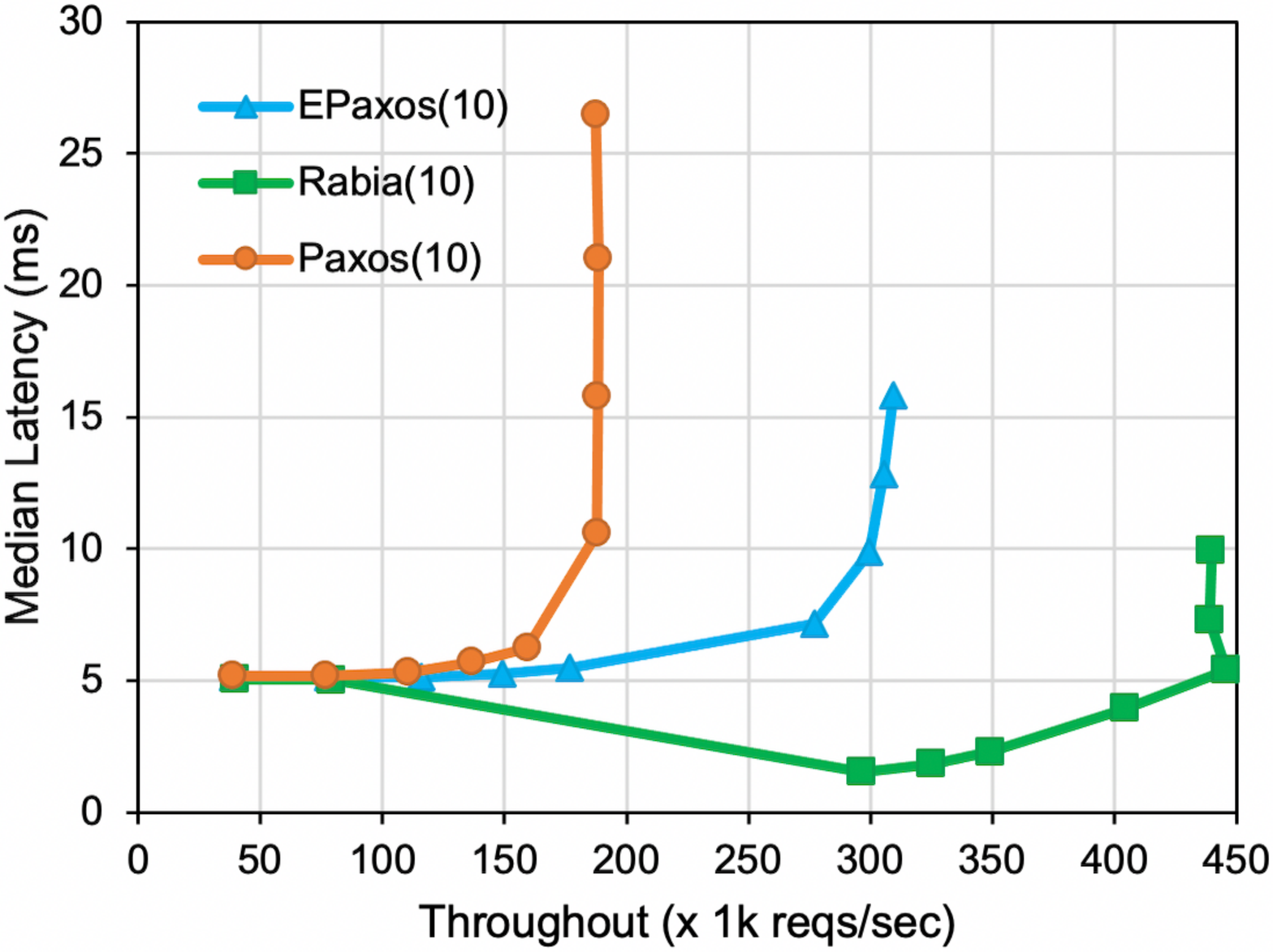}
  \caption{Median Latency. \\Same Zone with 3 Replicas.}
  \label{fig-sub:4a}
\end{subfigure}%
\begin{subfigure}{.24\textwidth}
  \centering
  \includegraphics[width=\textwidth]{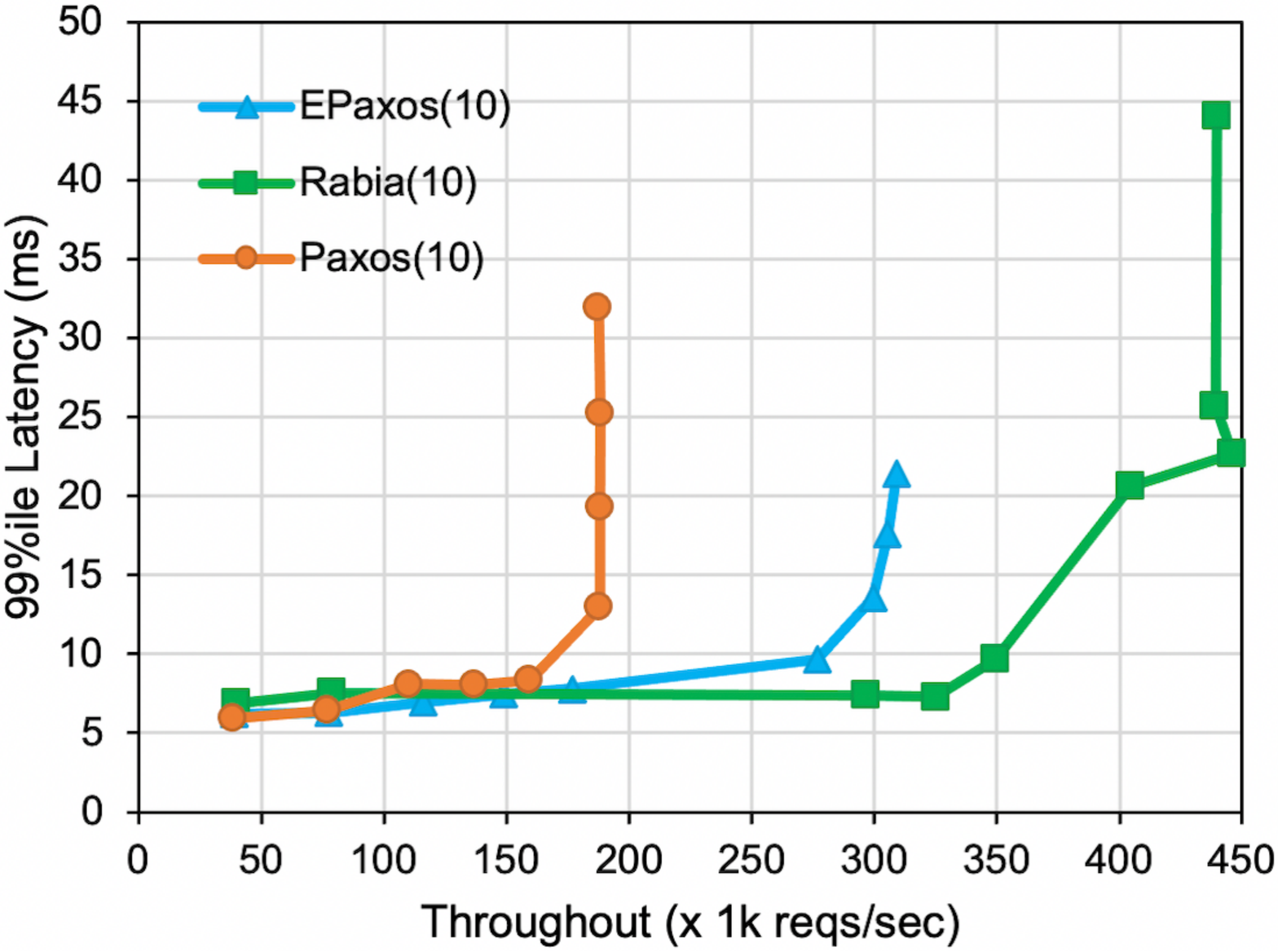}
  \caption{99th Percentile Latency. \\Same Zone with 3 Replicas.}
  \label{fig-sub:4b}
\end{subfigure}
\begin{subfigure}{.24\textwidth}
  \centering
  \includegraphics[width=\textwidth]{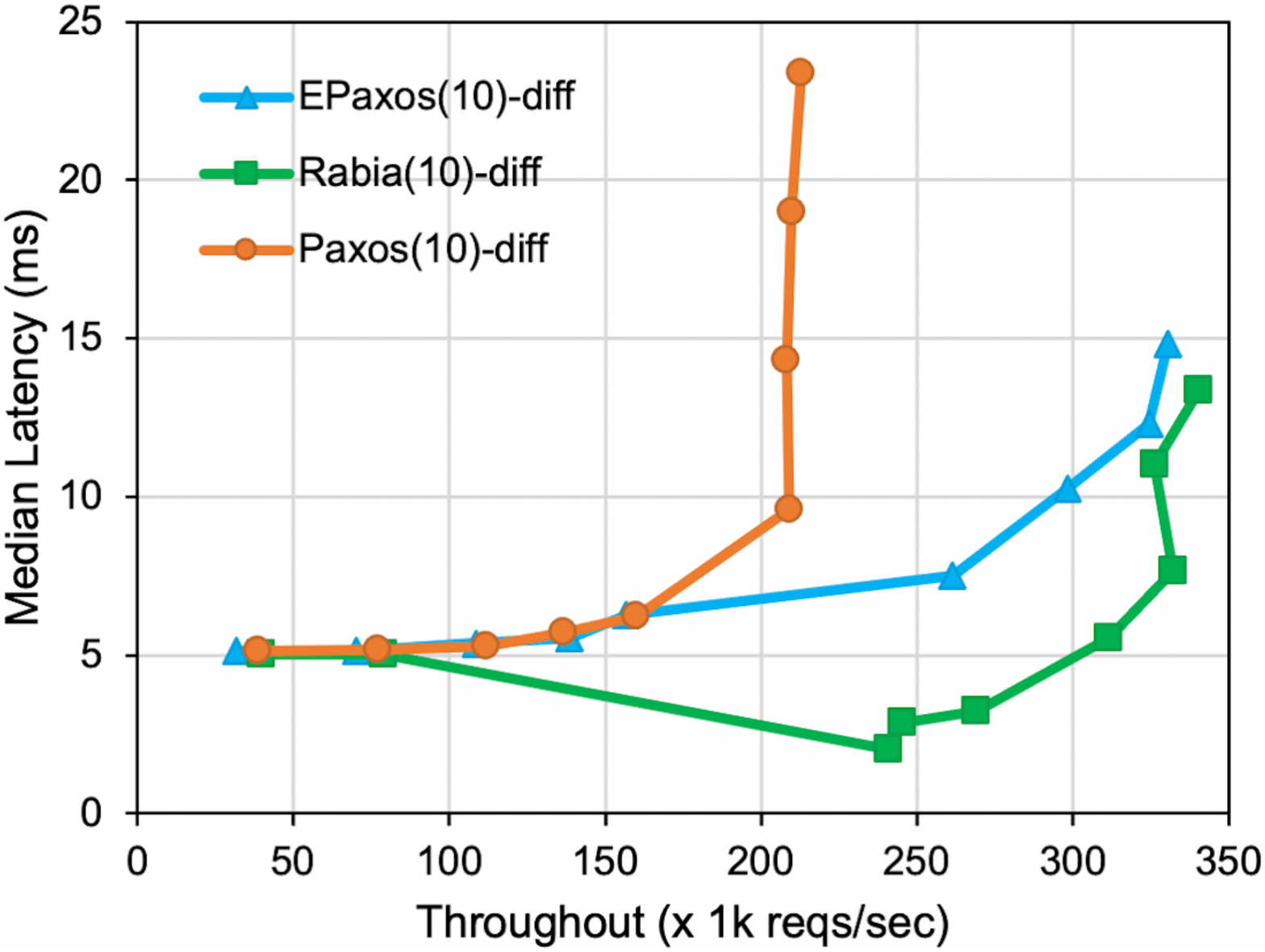}
  \caption{Median latency. \\Multi-Zones with 3 Replicas.}
  \label{fig-sub:4c}
\end{subfigure}
\begin{subfigure}{.24\textwidth}
  \centering
  \includegraphics[width=\textwidth]{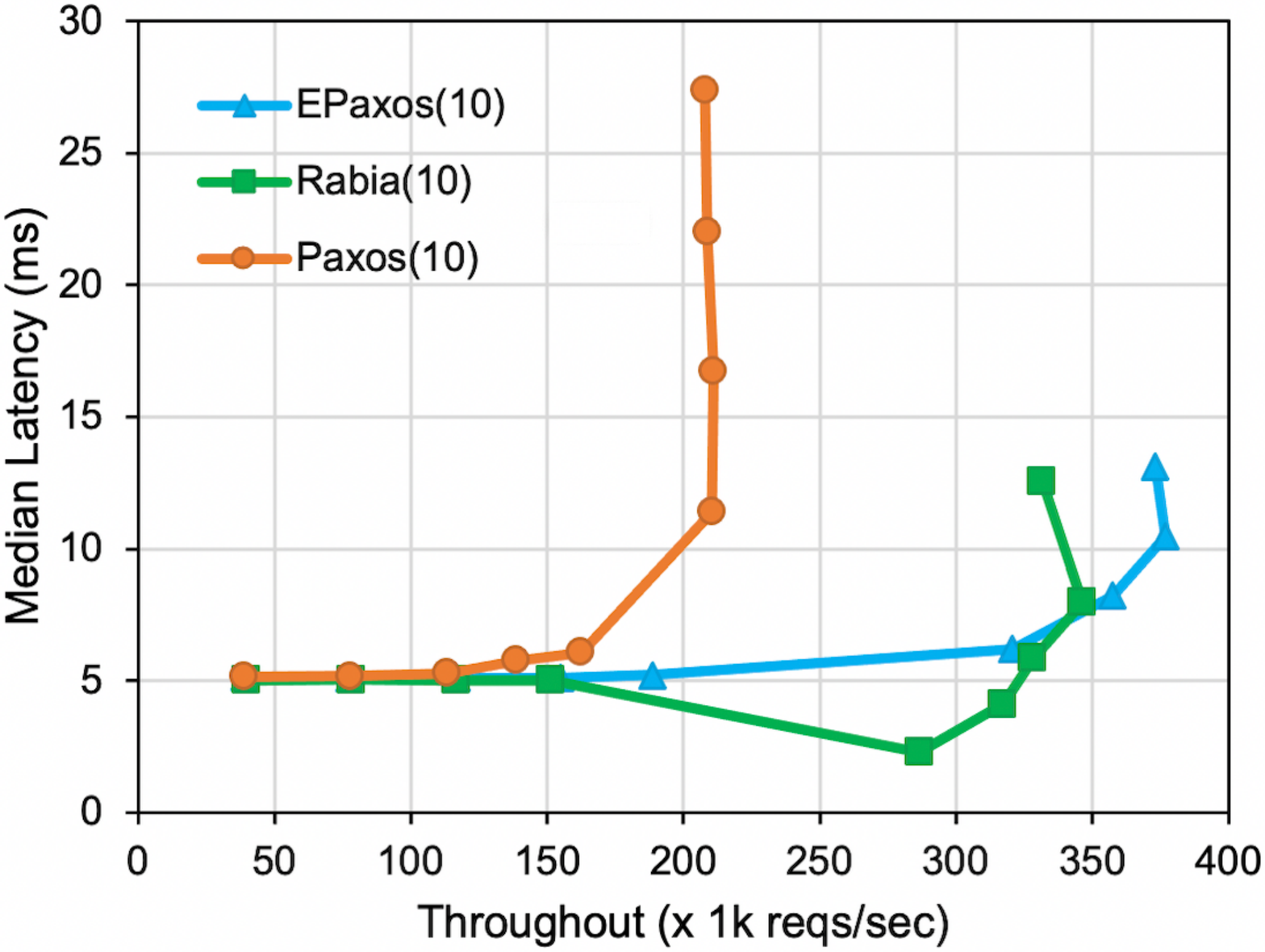}
  \caption{Median latency. \\Same Zone with 5 Replicas.}
  \label{fig-sub:4d}
\end{subfigure}%

%
\vspace{-10pt}
\caption{Throughput vs. Latency plots. The deployment consists of 3 or 5 4-CPU replica machines and 3 separate client machines. Both systems have client batch size $=10$, which is reported in brackets in the legend.}
\label{fig:throughput-latency}
\end{figure*}




\section{Performance Evaluation}
\label{s:evaluation}

We first evaluate \name{} by comparing it against two other systems: Multi-Paxos and EPaxos (with no conflicting requests) to understand the performance of the core consensus component. Then, we integrate \name{} with Redis (RedisRabia) and compare it with Redis. 
Multi-Paxos is a popular choice in production systems \cite{Chubby_OSDI06,Azure_SOSPI11,Paxos_live_PODC07}. EPaxos (with no conflict) is a state-of-the-art SMR system that achieves the best performance in our setting. It obtains high performance by achieving fast-path latency of two message delays, and each replica can have requests in a different order. 

The purpose of our evaluation is to understand the performance of the core algorithm with minimum optimization; hence, we choose Multi-Paxos and EPaxos that only use pipelining and batching. Highly optimized systems, e.g., Compartmentalized Paxos \cite{CompartmentalizedPaxos_VLDB21}, and systems that use specialized hardware like NOPaxos \cite{NOPaxos_OSDI16} could potentially achieve higher performance.

For EPaxos and Multi-Paxos, we use the implementation from \cite{EPaxos_SOSP13}. 
We implement a replicated key-value store that supports read (\textsc{Get}) and write (\textsc{Put}) operations based on \name. The structure follows closely to the one from \cite{EPaxos_SOSP13} so that we can perform a fair comparison. EPaxos uses its own GoBin library for serialization, which outperforms GoGo Protobuf,\footnote{\url{https://github.com/gogo/protobuf}} the library we used. However, we still choose GoGo Protobuf because it has native support to strings, which allows us to use Redis commands as inputs to \name{} directly. In the EPaxos evaluations, all requests are non-conflicting so that the achieved throughput is the maximum.

We evaluate our system on the Google Cloud Platform (GCP). 
Each server is on an e2-highmem-4 instance (Intel Xeon 2.8GHz, 4 vCPU, 32GB RAM) running Ubuntu-1604-xenial-v20201014. Clients run on one to three customized e2 machines, and each machine has 30 vCPUs and 120 GB RAM. All machines are deployed within a single availability zone us-east1-b (except for Figure \ref{fig-sub:4c}). Network bandwidth was measured in excess of 7.8 Gbits per second. The typical RTT is about 0.25 ms. 
The size of a client request is 16B (except for one experiment). We use these sizes because, for example, in Facebook's production TAO system \cite{Eiger_NSDI13,Facebook_TAO_ATC13},  50\% of requests have value field smaller than 16B. In~\cite{Facebook_Memcached_workload_sigmetrics12}, Facebook documents the workload of the Memcached deployment; in one of the systems, 40\% of requests is less than 11B.
In our setup, we observe no difference between different write ratios. For the data reported, the write ratio is 50\% for all three systems.

\vspace{3pt}
\noindent\textbf{Performance without Batching}:~
We first consider the case without any form of batching, i.e., both client batch and proxy batch are set to 1, and hence each slot contains one request. We test EPaxos and Paxos, with and without pipelining (labeled with (NP)). As described above, \name{} does not currently implement pipelining, 
and thus it processes one slot at a time. 
This closed-loop test consists of three replica machines and one client machine. We vary the load (number of clients) on each system until a maximum throughput is reached. The best performing number of  clients varies from two to six clients.\footnote{In this particular setting, EPaxos is  bottlenecked by dependency checking. The median latency of the check is 0.29 ms on the CPUs in GCP, which limits EPaxos's throughput. Hence, Paxos outperforms EPaxos in this evaluation.}

\begin{table}[ht]
\footnotesize
\begin{tabular}{|c|c|c|c|c|c|}
\hline
& \textbf{Rabia} & \textbf{EPaxos(NP)} & \textbf{EPaxos} & \textbf{Paxos(NP)}     & \textbf{Paxos }  \\ \hline
\textbf{Thpt}& 2458.56 & 2561.3 & 11480.1 & 1209.26  & 12993.07  \\ \hline
\textbf{M-Lat.}& 1.35 & 3.99 & 0.46  & 2.74 & 0.67  \\ \hline
\end{tabular}
\caption{\textbf{Performance without Batching.} (NP) indicates that a system has no pipelining. Throughput is represented as req/s, and median latency is measured in ms.}
\label{t:throughput}
\end{table}

\name{} is comparable with EPaxos (NP). This is because EPaxos spends more time in local computation, whereas \name{} has a higher communication overhead. 
EPaxos (NP) has a large median latency, compared to pipelined EPaxos, because it receives more messages before processing each slot, which increases the time for doing dependency check (that is proportional to the requests a command leader has seen). Naturally, the pipelined version outperforms \name{} by 5x. However, \name{} performs well without the pipelining optimization, considering the RTT is roughly 0.25ms. It is close to the maximum number of slots based on the theoretical analysis -- each slot takes 1.5 RTT on a fast path, which results into around 2667 slots per second. When batching is used, the effect of pipelining becomes less obvious, as serialization of client requests becomes the bottleneck. 

\vspace{3pt}
\noindent\textbf{Throughput vs. Latency}:~We next configure \name{}, EPaxos (with pipelining), and Paxos (with pipelining) to achieve its maximum throughput without reaching saturation and maintain a stable performance. Figure \ref{fig:throughput-latency} presents the throughput and latency numbers by increasing the number of concurrent closed-loop clients (20 -- 500). 
Interestingly, we found that an optimal configuration is different for each system. EPaxos and Multi-Paxos require proxy batching of size 1000 and 5000, respectively. \name{} runs best with small client batching and proxy batching. In this particular set of tests, \name{} uses proxy batching at 20. All systems use client batching size 10. 
The maximum batch size is 1000, 5000, and 300, for EPaxos, Paxos, and \name, respectively. These numbers are also used in \cite{EPaxos_SOSP13}. Following the best practice in \cite{EPaxos_SOSP13}, each system uses a timeout of 5ms to batch requests if the desired batch size is not reached.

Intuitively, different systems require different parameters to achieve the maximum throughput because each system has a different performance bottleneck. 
The reasons behind the choices of timeout and batch sizes, based on our experience, are:  (i) Paxos has a bottleneck at the leader, so it cannot initiate a slot too frequently; and (ii) \name{} does not implement pipelining, but distributes workload evenly at each replica; hence, it has to agree on one slot more quickly to obtain the maximum throughput. EPaxos falls in the middle. 

Figure \ref{fig-sub:4a} and Figure \ref{fig-sub:4b} show that \name{} has a small increase from the median latency to the 99th percentile latency when all three replicas are deployed in the same availability zone. However, if the system is overloaded, the 99th percentile latency becomes larger because of the increasing chances of NULL slots. Compared to EPaxos and Multi-Paxos, \name{} has a smaller drop in median latency because it uses a smaller proxy batch, and timeout is not triggered when there is a large enough number of clients.

Figure \ref{fig-sub:4c} shows that \name{} has a moderate degraded performance, around a 23\% drop, when deployed in multiple availability zones in the same region. The deployment has three replicas: one replica in us-east-1-b, one in us-east-1-c, and the other in us-east-1-d. All the clients are deployed in us-east-1-b. The average RTT across multiple zones increases from 0.25 to 0.4ms, and the variance becomes larger at 0.17ms. The limited drop in performance moving from one to multiple availability zones also shows that the network conditions that allow Rabia to reach consensus fast are practical, not stringent, and can be met even when Rabia is deployed beyond a single highly-connected network. Interestingly, EPaxos and Paxos have improved performance. This is because longer RTT actually reduces the communication burden on the leader node (in Paxos) and command-leader (in EPaxos).

Figure \ref{fig-sub:4d} presents the case with five replicas. Due to its $O(n^2)$ message complexity, \name{} has reduced throughput. The median latency remains small, compared to other systems.  EPaxos improves its performance, as also observed in \cite{EPaxos_SOSP13}, because when there are no conflicting requests, more requests can be handled concurrently with a larger number of replicas. Figures \ref{fig-sub:4c} and  \ref{fig-sub:4d} demonstrate that \name{} has high performance in the ideal case (same availability zone with $n=3$), but in other cases, the performance is still comparable to EPaxos.

\vspace{3pt}
\noindent\textbf{Varying Data Size}:~ Using the same configuration in Figure \ref{fig-sub:4a}, we also compare all three systems with key-value pairs with size 256B. EPaxos and Paxos suffer 71\% and 56\% reduction in throughput, whereas \name{} has less (47\%) reduction. Even though \name{} has a higher message complexity for each slot, EPaxos and Paxos rely on pipelining to obtain improved performance, which produces more messages in the same interval and therefore increases network utilization, reaching saturation faster.

\vspace{3pt}
\noindent\textbf{Integration with Redis}:~ In our integration, we use Redis to store the key-value pairs instead of a map stored in memory. RedisRabia utilizes Redis native $\textsc{MGet}$ and $\textsc{MPut}$ commands to process a batch of requests. 
Figure \ref{fig:redis} presents the comparison with native Redis-based systems: (i) synchronous-replication with one master and one replica (Sync-Rep (1)); (ii) synchronous-replication with one master and two replicas (Sync-Rep (2)); and (iii) RedisRaft, an experimental system by Redis Labs (Raft). We implement synchronous replication using Redis's $\textsc{WAIT}$ command on top of a native asynchronous replication in Redis (cluster mode). Note that data might be lost or become stale if the master fails in this type of replication. \name{} denotes a system without using Redis as the storage. In all the systems, we turned off any of the persistence options.

\begin{figure}[t]
    
    \includegraphics[width=.9\linewidth]{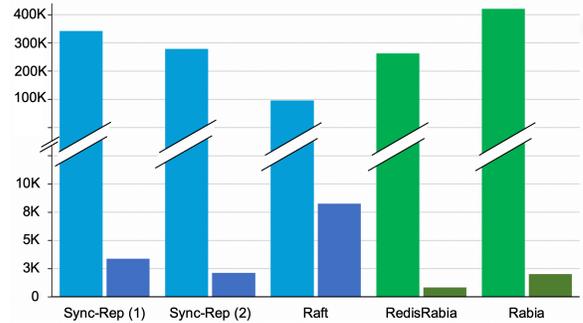}
    \vspace{-20pt}
    \caption{Throughput across different Redis integration. RedisRabia has around 1K req/s without batching, and Rabia has around 2K req/s without batching.}
     \label{fig:redis}
\end{figure}

Each system has two numbers, throughput with batching (left bar) and without batching (right bar). In all systems, we use 20 for client batching and 15 for proxy batching. The storage engine affects the performance of \name{} significantly because pipelining is not currently implemented, and any delay in completing a slot decreases throughput. Yet, RedisRabia is comparable with synchronous replication. The current implementation of RedisRaft is not optimizing throughput; hence, its throughput is sub-par except for the no-batching case due to its pipelining optimization. On the other hand, RedisRaft has better latency than all other systems, including \name-based ones. The reason for it 
is mainly the effect of pipelining, as also observed in Table \ref{t:throughput}.

\vspace{3pt}
\noindent\textbf{Internal Statistics and Network Stability Test}:~
The performance of \name{} depends on the underlying network, which affects two important internal statistics: (i) the message delays needed for Weak-MVC to terminate; and (ii) the percentage of NULL slots in the log. Appendix \ref{app:stat} presents a more thorough statistics in our evaluation. In summary, the maximum number of message delays we have observed is 15, which occurred 26 times out of 0.12 billion rounds, 
including the experiments that overload \name. ~96\% of slots use the fast path. For all the closed-loop tests, 2.22\% of slots are NULL, whereas for all the open-loop tests, 0.31\% of slots are NULL. We also conduct a simple network test in GCP and CloudLab \cite{CloudLab_ATC19} to measure how many messages a replica needs to receive
for all three replicas to have the same 
oldest message in an interval. This scenario ensures that all replicas have the same head in their respective priority queue ($PQ$). On average, replicas need to receive 3.1 to 3.9 messages. These results, along with the reported internal statistics, confirm that the network conditions of GCP have the stability needed to enable high performance in \name.

\section{Related Work}
\label{s:related}

\noindent\textbf{SMR}:~
State-machine replication is a popular mechanism for making a service fault-tolerant and highly available \cite{SMR_Schneider_ACM90}. 
Classical log-based SMR systems use a single leader to order requests, e.g., \cite{lamport1998part,lamport2001paxos,Raft_ATC14,Zab_DSN11,Mu_Aguilera_OSDI20}. The capacity of a single node limits these systems. Recent SMR systems adopt a multi-leader (or so-called leaderless) design that distributes the ordering responsibility evenly to increase performance (e.g., Mencius \cite{Mencius_Marzullo_OSDI08}, EPaxos \cite{EPaxos_SOSP13}, M$^2$Paxos \cite{M2Paxos_DSN16}, Caesar \cite{Caesar_DSN17}, and \textsc{Atlas} \cite{Atlas_Sutra_Eurosys20}). At the core of the multi-leader design \cite{Leaderless_theory_Sutra_DISC20}, there is a dependency graph to keep track of the dependency between requests so that any replica can decide the order of requests fast. As opposed to \name, both leader-based and multi-leader-based SMR systems require non-trivial effort \cite{Paxos_live_PODC07,Raft_CockroachDB16,Chubby_OSDI06} to integrate the core consensus algorithms with an auxiliary fail-over mechanism to recover from the failure of a leader or a command leader. Moreover, in the case of multi-leader-based SMRs, replicas also need to spend more computation in checking dependencies. These systems did not use randomized consensus.

\vspace{8pt}
\noindent\textbf{Randomized Consensus Algorithms}:~ Following Ben-Or's work, several theoretical papers \cite{Raynal_ISORC01,Chen_IPL09,Raynal_SRDS04,Aspnes_RandomConsensus_Survey_DC2003,ByzRandConsensus_Raynal_PODC14,Fault-tolerantBook_Raynal_Book18,RandByzConsens_Ittai_FC19,Rabin_FOCS83_randomizedByz,Raynal_TDSC05_randomizedByz} have improved it in multiple aspects, including tolerating Byzantine faults \cite{ByzRandConsensus_Raynal_PODC14,RandByzConsens_Ittai_FC19}, reducing the average number of message delays using a common coin \cite{Raynal_ISORC01,Fault-tolerantBook_Raynal_Book18,RandByzConsens_Ittai_FC19,Rabin_FOCS83_randomizedByz}, and taking multi-valued inputs \cite{Chen_IPL09}. Ben-Or's algorithm has been previously formally verified in several systems~\cite{MaricSB17, BertrandKLW19}. Pedone et al. \cite{Pedone_EDCC02_Weak-ordering} investigated two relaxed randomized consensus algorithms based on Ben-Or's algorithm \cite{Ben-Or_PODC83} and Rabin's algorithm \cite{Rabin_FOCS83_randomizedByz}, and also exploited the weakly ordering guarantees from the network layer to improve performance.  Crain \cite{Tyler_ByzRand_2020} presented an evaluation of his binary Byzantine consensus algorithm. The implementation only takes binary input and generates a binary output. These works focused on the single-shot consensus, and did not present an approach to develop a practical SMR system, as Rabia does.



\vspace{8pt}
\noindent\textbf{Practical Randomized Byzantine Fault-tolerance Systems}:~Randomization has been recently used in Byzantine Fault-tolerance (BFT) systems. Cachin et al. proposed a coin-flipping protocol based on the Diffie-Hellman problem, which is then used to build a Byzantine consensus protocol with constant message delays in expectation and has message and communication complexities close to the optimum \cite{Cachin_JCrypto05_randomBFT}. Cachin et al. \cite{Cachin_OPODIS16_randomBFT} introduced a protocol for cryptographically secure  randomness generation that allows users to deploy non-deterministic applications on top of the proposed BFT systems. Miller et al. presented HoneyBadgerBFT \cite{HoneyBadger_CCS16}, which also identified that randomization can be used to improve performance in SMR systems. The key technique is a randomized atomic broadcast protocol that tightly integrates random selection and encryption.

Another popular application of randomized consensus is Blockchain. For example, the Proof-of-Work of Bitcoin \cite{Bitcoin_Nakamoto2009}  and Proof-of-Stake of Algorand \cite{Algorand_Micali_TCS2019,Algorand_Micali_SOSP2017} can both be viewed as randomized consensus algorithms. These BFT systems focus on cryptographical guarantees \cite{Cachin_JCrypto05_randomBFT,Cachin_OPODIS16_randomBFT}, permissionless settings \cite{Bitcoin_Nakamoto2009,Algorand_Micali_TCS2019,Algorand_Micali_SOSP2017}, or wide-area networks \cite{Bitcoin_Nakamoto2009,Algorand_Micali_TCS2019,Algorand_Micali_SOSP2017,HoneyBadger_CCS16}. These aspects are significantly different from our target scenario.  Therefore, their techniques, including the ways that randomization is used, are different from \name{}'s. Specifically, these systems do not use relaxed consensus (with weak validity) for improving performance. 

\section{Conclusion}
\label{s:summary}

We present the design and implementation of \name{}, a simple SMR system that achieves high performance in favorable settings (e.g., same availability zone with $n=3$) and comparable performance in more general cases within a single datacenter. By using randomized consensus as opposed to a deterministic one, \name{} does not need any fail-over protocol and supports trivial log compaction. Our evaluation study confirms that in commodity networks, the belief that randomized consensus implementations do not provide competitive performance can be challenged.

\section*{Acknowledgment}
The authors thank our shepherd Manos Kapritsos and all anonymous reviewers and artifact reviewers for their important comments. This material is based upon work supported by the National Science Foundation under Grant No. CNS-1816487 and CNS-2045976. Early evaluation and artifact evaluation of this work were obtained using the CloudLab testbed \cite{CloudLab_ATC19} supported by the National Science Foundation. The authors would also like to thank Matthew Abbene and Andrew Chapman's help on Redis integration, and Rachel Trickett's help on artifact evaluation.

\bibliographystyle{plain}
\bibliography{references,DS_system,n-str,n,n-conf} 
\appendix
\noindent {\Large \textbf{Appendix}}

\section{No-op in Paxos}
\label{app:paxos-noop}

Paxos and variants have a concept similar to our usage of NULL values, namely no-op command. For example, in Multi-Paxos, a proposer/leader might not know the agreed value of a certain set of slots. Such a scenario is common for a newly elected leader. In this case, the proposer sends a no-op to acceptors to learn value(s) that have been previously proposed by other proposers and accepted by a quorum. This allows the new proposer to fill in any missing slots (holes in the log) in the sequence so that existing requests can be executed. In other words, in Multi-Paxos, each slot still contains a request. In contrast, a slot might contain a NULL value (no request) in Rabia.

\section{Local Computation Time in EPaxos}
\label{app:EPaxos-local-time}

Table \ref{t:EPaxos-time} presents the median function computation time in EPaxos (no conflict requests). We measure the time for each function in a closed-loop with 100 clients and varying batch sizes. We only report functions beyond 0.1ms. The bottom row presents the total local computation time a replica takes to commit a slot.

\begin{table}[ht]
\centering
\footnotesize
\begin{tabular}{l | l | l | l}
\textbf{Function} & \textbf{Batch = 80} & \textbf{Batch = 10} & \textbf{Batch = 1}\\\hline
handlePropose & 0.42 & 0.2 & 0.06 \\
handlePreAcceptReply & 0.42 & 0.19 & 0.06 \\
handlePreAcceptOK & 0. 44& 0.57 & 0.11\\
handleAcceptReply & 0.42 & 0.11 & 0.04\\\hline
total computation time & 1.80 & 1.12 & 0.29
\end{tabular}
\caption{Median EPaxos function computation time (ms).}
\label{t:EPaxos-time}
\end{table}

\section{Practical Tail Latency Reduction}
\label{app:tail-latency-reduction}

We describe two practical approaches to reduce tail latency further.

\begin{itemize}[nosep]
\item\textit{Failure detector}: As the example in Section \ref{s:weak-MVC} demonstrates, the cause of long tail latency is when replicas take a different branch of the if statements in Line 23 of Algorithm \ref{algo:Weak-MVC}. We can use an eventually correct failure detector \cite{phi_failureDetector_SRDS04} to avoid such a scenario. If an instance of Weak-MVC takes unusually long, then replicas will wait for $\stateRabia$ and $\vote$ messages from all the live replicas indicated by the failure detector. This allows live replicas to receive a consistent set of messages, which will enable replicas to take the fast path and reach an agreement in that phase.
\item\textit{Freeze time}: One major reason that replicas do not use the fast path is that they use different proposals as input to the Weak-MVC. We can avoid this case by choosing the proposal more carefully. At Line 2 of Algorithm \ref{algo:SMR-server}, if the $proposal_i$ has a timestamp that is too close to the current time, then replica $N_i$ should wait for a small amount of time, called \textit{freeze time}, and check $PQ_i$ again to see if there is a request that is older than the current $proposal_i$. If so, then $N_i$ switches the proposal; otherwise, a small freeze time would also give other replicas more time to receive $proposal_i$. Consequently, with waiting, replicas will likely to have more similar $PQ$'s; hence, freeze time increases the chance of hitting the fast path and hence, reduces the probability of a long latency.
\end{itemize}

\section{Service Availability under Failures}
\label{app:failure}

\begin{figure}[ht]
    \includegraphics[width=.9\linewidth]{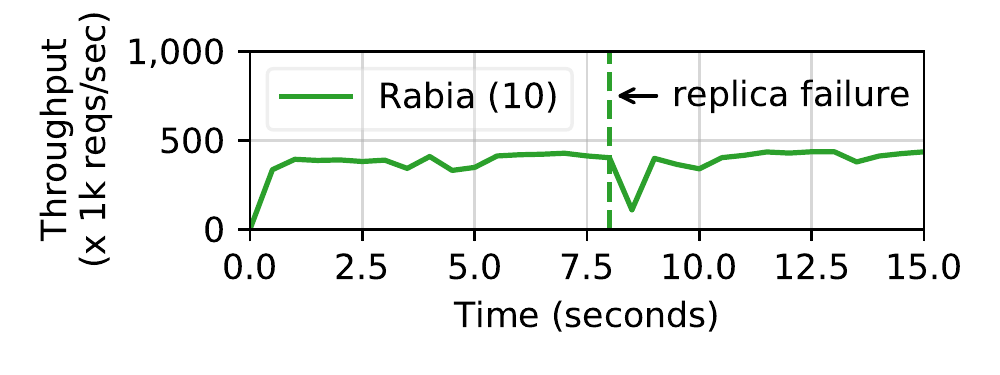}
    
    \caption{Throughput with one crashed replica. The lowest throughput is around 101K reqs/second.}
     \label{fig:recover}
\end{figure}

Figure \ref{fig:recover} shows the evolution of throughput in a 3-replica system with 60 closed-loop clients. The system experiences a single crashed replica in the execution. The throughput drops because clients use a timeout-based mechanism to detect a replica failure and then switch to another replica as described in Section \ref{s:discussion}. The throughput recovers and stabilizes after all the clients complete the switch to the live replicas.

\section{Internal Statistics and Network Stability Test}
\label{app:stat}
Two important numbers affect the performance of \name{}: (i) the message delays needed for Weak-MVC to terminate; and (ii) the percentage of NULL values in the log. Table \ref{t:delay} presents the percentage of message delays for all the evaluations that we have done. The maximum number of message delays we have observed is 15, which occurred 5 times throughout our evaluation. For all the closed-loop tests, 2.22\% of slots are NULL, whereas for all the open-loop tests, 0.31\% of slots are NULL.  For all the closed-loop tests deployed across multiple zones, 2.09\% of slots are NULL.

\begin{table}[ht]
\footnotesize
\begin{tabular}{|c|c|c|c|c|}
\hline
                              \textit{Message delays}            & $3$       & $5$      & $7$      & $9$ -- $15$  \\ \hline
Open-loop                                                              & 99.58\% & 0.37\% & 0.04\% & 0.01\% \\ \hline
Closed-loop (all)                                                      & 96.81\% & 2.76\% & 0.39\% & 0.04\% \\ \hline
\begin{tabular}[c]{@{}c@{}}Closed-loop \\ (same zone)\end{tabular}     & 96.90\% & 2.78\% & 0.28\% & 0.04\% \\ \hline
\begin{tabular}[c]{@{}c@{}}Closed-loop\\ (three zones)\end{tabular} & 96.77\% & 2.75\% & 0.44\% & 0.04\% \\ \hline
\end{tabular}
\caption{Message Delays of Weak-MVC}
\label{t:delay}
\end{table}

\vspace{3pt}
\noindent\textbf{Network Stability Test}:~ In order to empirically assess the network stability in our experiments, we devise a test to quantify whether replicas can receive a consistent set of messages in a similar setup. In our test, three senders send messages concurrently to three receivers every 0.3ms (roughly an RTT in GCP). We measure how many consecutive messages the three receivers need to receive in order to collect all three messages sent by the senders in the same time interval. We collect this number multiple times during multiple days on GCP. Removing few outliers, the mean value fluctuates between 3.1 to 3.9, and the 95th percentile constantly stays around 5. 
These results, along with the reported internal statistics, confirm that the network condition of GCP satisfies our stability desire to enable high performance in \name{}.



\end{document}